\def\acm{1}
\def\llncs{0}
\def\ieee{0}

\if\acm1
\documentclass[acmsmall,nonacm]{acmart}
\fi

\if\llncs1
\documentclass[runningheads,citeauthoryear]{llncs}
\usepackage[margin=1in]{geometry}
\fi

\if\ieee1
\documentclass[conference,compsoc]{IEEEtran}
\fi

\def\submissionnumber{150}

\ifnum\pdfstrcmp{\jobname}{main}=0
\def\anonymize{0}
\def\release{1}
\def\shrink{0}
\fi

\ifnum\pdfstrcmp{\jobname}{draft}=0
\def\anonymize{0}
\def\release{0}
\def\shrink{0}
\fi

\ifnum\pdfstrcmp{\jobname}{anonymized}=0
\def\anonymize{1}
\def\release{1}
\def\shrink{1}
\fi

\ifnum\pdfstrcmp{\jobname}{output}=0
\def\anonymize{0}
\def\release{1}
\def\shrink{0}
\fi

\usepackage{booktabs} % For formal tables
\usepackage[ruled]{algorithm2e} % For algorithms

\if\shrink1
  \usepackage[subtle]{savetrees}
\fi

\usepackage[utf8]{inputenc}

\if\llncs1
  \def\safeqed{\qed}
  \usepackage{amssymb}
  \newtheorem{observation}{Observation}
  \newtheorem{assumption}{Assumption}
\fi
\if\acm1
  \usepackage{amsthm}
  \def\safeqed{}
\fi

\if\ieee1
  \usepackage{amsthm,amssymb}
  \usepackage{cite}
  \def\safeqed{}

  \newtheorem{lemma}{Lemma}
  \newtheorem{definition}{Definition}
  
  \newtheorem{assumption}{Assumption}
  \newtheorem{theorem}{Theorem}
\fi

\usepackage{amsmath}
\usepackage{breqn}
\usepackage{hyperref}
\usepackage[T1]{fontenc}

\usepackage{subcaption}

\SetAlFnt{\small}
\SetAlCapFnt{\small}
\SetAlCapNameFnt{\small}
\SetAlCapHSkip{0pt}
\if\acm1
  \setcopyright{none}

  \setcitestyle{acmnumeric}
\fi

\if\llncs1

\fi

\title{Fair Ordering in Replicated Systems via Streaming Social Choice}

\newcommand{\geoffemail}{geoff.ramseyer@cs.stanford.edu}
\newcommand{\ashishemail}{ashishg@stanford.edu}

\newcommand{\acmauthors}
{
\if\anonymize1
\author{Submission \submissionnumber}
\else
\author{Geoffrey Ramseyer}
\affiliation{%
  \institution{Stanford University}
  \city{Stanford}
  \state{California}
  \postcode{94305}
  \country{USA}}
\email{\geoffemail}
\author{Ashish Goel}
\affiliation{%
  \institution{Stanford University}
  \city{Stanford}
  \state{California}
  \postcode{94305}
  \country{USA}}
\email{\ashishemail}
\fi
}

\if\acm1
\acmauthors{}
\fi

\newcommand{\llncsauthors}
{
\if\anonymize1
\author{Submission \submissionnumber}
\authorrunning{Submission \submissionnumber}
\institute{}
\else
\author{Geoffrey Ramseyer\inst{1}\and Ashish Goel\inst{2}}
\authorrunning{G. Ramseyer and A. Goel}
\institute{Stanford University, Stanford, CA 94305, USA
\email{\geoffemail} \and 
Stanford University, Stanford, CA 94305, USA
\email{\ashishemail}}
\fi
}

\newcommand{\ieeeauthors}
{
\if\anonymize1
\author{%
  \IEEEauthorblockN{Submission \submissionnumber}
}
\else

\author{%
  \IEEEauthorblockN{Geoffrey Ramseyer}
\IEEEauthorblockA{Stanford University \\
\geoffemail{}}
\and
\IEEEauthorblockN{Ashish Goel}
\IEEEauthorblockA{Stanford University\\
\ashishemail{}}
} 
\fi

}

\if\llncs1
\llncsauthors{}
\fi

\if\ieee1
\ieeeauthors{}
\fi

\if\release1
\newcommand\XXX[1]{}
\else
\newcommand\XXX[1]{\begingroup \bfseries\color{red} XXX: #1\endgroup}
\fi

\newcommand{\tx}{\tau}
\newcommand{\txp}{\tau^\prime}
\newcommand{\txpp}{\tau^{\prime\prime}}
\newcommand{\before}{\prec}

\usepackage[bibliography=common]{apxproof}

\newcommand{\aaaielide}[2]{#1}

\newtheoremrep{theorem}{Theorem}[section]
\newtheoremrep{lemma}[theorem]{Lemma}
\newtheoremrep{observation}[theorem]{Observation}
\newtheorem{assumption}[theorem]{Assumption}
\newtheoremrep{corollary}[theorem]{Corollary}

\begin{document}

% Abstract. Note that this must come before \maketitle.

\if\acm1
\begin{abstract}

How can we order transactions ``fairly'' in a replicated state machine (of which today's blockchains are a prototypical example)?
In the model of prior work \cite{themis,kelkar2020order,cachin2022quick,vafadar2023condorcet,kiayias2024ordering},
each of $n$ replicas observes transactions
in a different order, and the system aggregates
these observed orderings into a single order.  We argue that this problem
is best viewed directly through the lens of the classic preference aggregation problem of social choice theory 
(instead of as a distributed computing problem),
in which rankings on candidates are aggregated into an election result.

Two features make this problem novel and distinct.
First, the number of transactions is unbounded, and an ordering must be defined
over a countably infinite set.  And second, decisions must be made quickly and with only partial information.
Additionally,
some faulty replicas might misreport their observations;
the influence of faulty replicas on the output should be well understood.

Prior work studies a ``$\gamma$-batch-order-fairness'' property,
which divides an ordering into contiguous batches.
If a $\gamma$ fraction of replicas receive a transaction $\tx$ before another transaction $\txp$, then
$\txp$ cannot be in an earlier batch than $\tx$.  
This definition holds vacuously, so we strengthen it 
to require that batches have minimal size,
while accounting for faulty replicas.

This lens gives a protocol with 
both strictly stronger fairness and better liveness properties
than prior work.
We specifically adapt the Ranked Pairs \cite{tideman1987independence}
method to this streaming setting.
This algorithm can be applied on top of any of the communication protocols (in various network models)
of prior work for immediate liveness and fairness improvements.
%We study the Ranked Pairs method.
%Analysis of how missing information moves through the algorithm
%allows our streaming version to know when it can output a transaction.
%Deliberate construction of a tiebreaking rule ensures our algorithm outputs a transaction after a bounded time in a synchronous network,
%although our protocol can be applied, for immediate liveness improvements, 
%on top of any of the communication protocols of prior work (in various network models).
Prior work relies on a fixed choice of $\gamma$ and a bound on the number of faulty replicas $f$,
but we show that Ranked Pairs satisfies our definition
for every $\frac{1}{2}<\gamma\leq 1$ simultaneously and for any $f$, where fairness
guarantees degrade as $f$ increases.

\end{abstract}
\fi

% Title page for title and abstract only.
%\begin{titlepage}

\maketitle

\if\llncs1
\begin{abstract}

\end{abstract}
\fi

\if\ieee1

\begin{abstract}

\end{abstract}

\fi

%\end{titlepage}
\section{Introduction}

We study the problem of ordering transactions in the widely-used replicated state machine architecture (e.g., 
\cite{ongaro2014search,lamport2001paxos,yin:hotstuff,oki1988viewstamped}).
In the standard setting, each of a set of $n$ distinct \textit{replicas}
maintains a copy of a \textit{state machine}.
Replicas
communicate to agree on a totally-ordered log of transactions $\tx_1 \before \tx_2\before\tx_3\before\dots$,
and then each replica applies the transactions in this order to its local copy of the state machine.
An (unbounded) set of clients create and broadcast new transactions.
This architecture famously underlies many of today's blockchains,
but also underlies many traditional (centralized) services, such as bank infrastructure.
%A classic example is a banking application, in which
%the state machine maintains a list of user account balances,
%and a transaction transfers money from one user to another.

This architecture requires that all replicas apply transactions in the same order.
In a bank example,
a client $A$ with a \$1 balance might create two transactions, one which sends a \$1 payment to client $B$
and one which sends a \$1 payment to client $C$.  Only one of payments can execute successfully 
(assuming $A$'s balance is not allowed to overdraft) 
so if different replicas apply these transactions in different orders,
then they will disagree on the balances of clients $B$ and $C$.

There are many communication protocols through which replicas can agree on \textit{some} total order 
(solving the ``total order broadcast'' problem; \citet{broadcastsurvey} gives a survey).
Yet in many systems, most notably within today's public blockchains,
significant financial value can be derived from ordering transactions in specific ways \cite{daian2020flash}.
These systems must therefore agree on not just any ordering but an ``optimal'' one, for some notion of optimal.

The key observation for this work is that this problem is a novel streaming variation of the classic
preference aggregation problem of social choice theory \cite{hagele2001lulls,llull,colomer2013ramon,condorcet1785essay,arrow1950difficulty}.
Prior work on this problem \cite{kelkar2020order,themis,cachin2022quick,aequitaspermissionless} 
use very different network communication protocols, but ultimately each produces agreement on a ``vote'' from each
replica on the ordering of a set of transactions (in social choice parlance, a ``ranking'' on a set of ``candidates'').
Agnostic to the choice of network protocol,
how should a system aggregate a set of proposed orders (``rankings'') into a single, total order?

Two key differences separate this problem from the classic social choice setting.
First, the number of transactions (``candidates'') is countably infinite
(as clients can continually send new transactions for an unbounded time).
And second, the system must produce the output ranking in a streaming fashion.
It cannot wait to see the reported orders over the entire (infinite) set of
transactions before making a decision on the relative ordering of two transactions.
Instead, the system must produce as output an append-only, totally-ordered log of transactions (to be given to the state machine).
Furthermore, the system should ideally minimize the delay between when a client sends a transaction
and when that transaction is appended to the output log (i.e. the system needs ``liveness'').

While we believe this problem to be interesting in its own right,
studying it directly through the lens of classic social choice theory,
enables us to develop an ordering algorithm
 with both much stronger order ``fairness'' guarantees (the key property studied in prior work, discussed below)
 and stronger liveness guarantees
than all of the prior work.  This algorithm could be deployed on top of any of the network protocols
of prior work.

One additional consideration is that replicas might strategically adjust their reported orderings.
Prior work \cite{kelkar2020order,themis,cachin2022quick} writes that ``honest'' replicas
must report transactions in the order in which they arrive over a network,
and any other behavior is ``faulty.''
One desirable property of an order aggregation rule is that the influence of a (colluding) subset of faulty replicas is 
precisely bounded. 
%Mirroring the prior work, we write here that it is ``faulty'' behavior for a replica
%to report anything other than the order in which it observes transactions arriving over the network.

There is a wide body of social choice literature on this underlying aggregation problem.
Our goal with this work is to demonstrate an application for our streaming version of the problem,
and to demonstrate the value of using social choice results in this application
by targeting the precise desiderata raised in prior work.  
There are many other natural desiderata,
notions of fairness,
and aggregation rules that may be practically useful. 
The application of social-choice style aggregation rules in this streaming setting poses a number of interesting open questions. 
For example, if replicas have distinct financial motivations, or accept bribes from clients to order transactions in specific ways,
is there an aggregation rule that maximizes (perhaps approximately) social welfare?
For a reader coming from social choice theory,
``transaction'' could be replaced by ``candidate'', and ``an ordering vote'' by ``a ranking.''

\subsection{Our Results}

%To the best of our knowledge, this work is the first to directly pose the question
%as a streaming version of the classic preference aggregation problem.
We first observe that for an ordering method to be well-defined on a countably infinite set,
it suffices for there to exist a ``monotonic'' and ``asymptotically live'' algorithm that implements the ordering method
on finite, initial segments of the input.  This algorithm must not undo a decision
when its (finite) input is extended (''monotonicity'') and must eventually output every transaction (''asymptotic liveness'').
The core challenge in building such an algorithm is in determining when an algorithm has enough 
information to make a decision that is consistent with its hypothetical output on any extension of its input.

Prior work of \citet{kelkar2020order} defines a notion of ``$\gamma$-batch-order-fairness'' (reproduced here as Definition \ref{defn:gammabatch}).
The output ordering is divided into batches, and if $\gamma n$ replicas receive a transaction $\tx$ before another transaction $\txp$,
then $\tx$ cannot be in a later batch than $\txp$.  However, transactions within a batch can be ordered arbitrarily.
Unfortunately, this definition is vacuously satisfiable by an arbitrarily large batch.  We first strengthen this definition
in \S \ref{sec:orderdefn} to precisely capture the intuition that batches should be \textit{minimal}.
%We first strengthen in \S \ref{sec:orderdefn} this definition to precisely captures the intuition
%that batches should be \textit{minimal}.
%Minimality is essential to a useful definition fairness,
%as otherwise, a protocol vacuously satisfies the definition by placing all transactions in the same batch.
Minimality must be defined carefully in the presence of faulty replicas;
%Our definition is, in fact, required in order to reason about order fairness; 
we show that
$\gamma$-batch-order-fairness cannot be satisfied simultaneously with exact notions of batch minimality 
%(e.g. as defined in \cite{themis})
and faulty replicas.

We call our definition $(\gamma, \delta)$-minimal-batch-order-fairness (Definition \ref{defn:minimalbatch}).
Simply put, if a $\frac{1}{2}<\gamma\leq 1$ fraction of replicas vote for $\tx$ before $\txp$ (henceforth; $\tx \before \txp$)
 then the output ordering should include $\tx\before\txp$
unless there is a sufficiently strong reason to put $\txp\before\tx$---that is, a sequence of transactions 
$\left(\tx_1,\tx_2,\ldots,\tx_{k-1},\tx_k\right)$, with $\tx_1=\tx$ and $\tx_k=\txp$,
where at least a $\gamma-2\delta$ fraction of replicas vote for $\tx_i\before\tx_{i+1}$ for $1\leq i \leq k-1$.
Motivating this definition is the fact that a $\delta$ fraction of faulty replicas can, at most, 
increase or decrease the fraction of replicas that report $\tx\before\txp$ by $\delta$.

As an example, given $f$ faulty replicas
out of $n$ total replicas and a fixed $\gamma$,
the protocol of \citet{kelkar2020order}
 satisfies
$(\gamma, \frac{f}{n})$-minimal-batch-batch-order-fairness (Lemma \ref{lemma:aeq_batch}).
Unfortunately, that protocol relies on an explicit choice of a parameter $\gamma$ and an explicit bound on $f$, and
it is not at all clear whether a higher or lower parameter $\gamma$ gives a stronger guarantee (\S \ref{apx:orderexamples} gives an example).
By contrast, the algorithm we give here achieves $(\gamma, \frac{f}{n})$-minimal-batch-order-fairness for every $\frac{1}{2}<\gamma\leq 1$
\textit{simultaneously},
and for any number of faulty replicas.  Fairness guarantees smoothly degrade as the number of faulty replicas increases.
\footnote{Our algorithms do not require knowledge of a bound on the number of faulty replicas.
This notion of ``faulty'' replica is distinct from notions of ``faulty'' in classical Byzantine-fault tolerant communication protocols, such as \cite{castro1999practical}.}
%This number may be separate from a bound on the number of faulty replicas
%in an underlying communication protocol.  Common Byzantine-fault tolerant protocols require $n\geq 3f+1$.\XXX{cite of like liskov pbft}}

This work adapts the ordering known as Ranked Pairs \cite{tideman1987independence} to our streaming setting.
Informally, Ranked Pairs operates on a weighted, directed graph, where each vertex is a transaction
and an edge from $\tx$ to $\txp$ has weight equal to the fraction of replicas that report $\tx\before\txp$.
Prior work \cite{kelkar2020order,cachin2022quick}, by contrast, operates on the unweighted, directed graph that is derived
by taking the weighted graph which our work uses,
removing all edges below weight $\gamma$, and dropping the weights on the remaining edges.  This technical difference
is the key to enabling our stronger results.

Analysis of this directed graph gives an important structural lemma about what pieces of information
Ranked Pairs uses to determine its output.  This lemma allows a streaming version of Ranked Pairs to carefully track how uncertainty about as-yet-unseen transactions
propagates throughout Ranked Pairs's algorithm.
Ranked Pairs iterates over edges in order of edge weight.  
An interesting observation is that the tiebreaking rule in this edge ordering critically determines overall liveness.
A worst-case (fixed) rule could cause our streaming algorithm to never output any transactions.
However, we show in \S \ref{sec:ordering:implementation} how to construct a tiebreaking rule in a streaming fashion
that guarantees not only asymptotic liveness but also an explicitly bounded liveness delay.
If, for example, at most $\Delta$ time elapses between
when a transaction $\tx$ is sent and when every honest replica votes on it, for example, 
 our algorithm outputs $\tx$ after a delay of at most $O(n\Delta)$.
 This can be reduced to $O(k\Delta)$, by rounding edge weights to the nearest $\frac{1}{k}$---
although this rounding
reduces the fairness guarantee to $(\gamma, \frac{f}{n} + \frac{1}{2k})$-minimal-batch-order-fairness (again for every $\frac{1}{2}<\gamma\leq 1$ simultaneously,
and any $f$).

%Using a naive deterministic rule,
%our streaming algorithm might never output any transactions.
%However, we then show in \S \ref{sec:ordering:implementation} that deliberate construction of the tiebreaking rule
%guarantees asymptotic liveness.

%Finally, a synchronous network assumption combined with our tiebreaking guarantees 
%that our algorithm outputs each transaction after a finite amount of time.
%If at most $\Delta$ time elapses between when a transaction $\tx$ is sent and when every honest replica votes on it, for example, 
% our algorithm outputs $\tx$ after a delay of at most $O(n\Delta)$.
% This can be reduced to $O(k\Delta)$, by rounding edge weights to the nearest $\frac{1}{k}$---
%although this rounding
%reduces the fairness guarantee to $(\gamma, \frac{f}{n} + \frac{1}{2k})$-minimal-batch-order-fairness
% (again for every $\frac{1}{2}<\gamma\leq 1$ simultaneously,
%and any $f$).

\section{Preliminaries and System Model}
\label{sec:ordering:model}

We consider a model in which there are $n$ \textit{replicas} cooperating to develop a total ordering
of transactions. 
Transactions are received over the network from \textit{clients}.
We say that the ordering in which a replica receives transactions is that replica's observed ordering.

\begin{definition}[Ordering Preference]
\label{defn:recv}
An \textit{Ordering Preference} on a (finite or countably-infinite) set of transactions is a total ordering
\footnote{
Isomorphic to a subset of $\omega$}
$\hat{\sigma_i}=\left(\tx_{i_1}\before\tx_{i_2} \before \tx_{i_3}\before \ldots\right)$
\end{definition}

However, at any finite time, each replica can have received only a finite number of transactions.
Replicas periodically submit these finite ``ranking votes'' to a ranking algorithm.

\begin{definition}[Ordering Vote]
A replica $i$ submits to a ranking algorithm
an \textit{ordering vote} 
on a set of $k$ transactions, $\sigma=\left(\tx_{i_1},...\tx_{i_{k}}\right)$.
(where $\tx_{i_{j_1}} \before \tx_{i_{j_2}}$ for $j_1 < j_2$).

We say that $\sigma^\prime$ extends $\sigma$ if $\sigma$ is an initial segment of $\sigma^\prime$
(as a convention, $\sigma$ extends itself).
\end{definition}

\begin{definition}[Ranking Algorithm]
A deterministic ranking algorithm $\mathcal{A}(\sigma_1,\ldots,\sigma_n)$ takes as input an ordering vote from
each replica and outputs an ordering $\sigma$ on a subset of the transactions in its input.
\end{definition}

The output need not include every transaction in the input.

We say that a replica is \textit{honest} if its true, observed ordering always extends its ordering vote,
and if whenever it submits a vote, it contains all transactions that the replica has observed at that time. 
Otherwise, the replica is \textit{faulty}.  
We denote the number of faulty replicas as $f$ (out of $n$), and do not assume a bound or knowledge of $f$.
%We assume that at most $f$ of the $n$ replicas are faulty, although our algorithms do not require knowledge of $f$.
\footnote{
	Of course, $(\gamma, \frac{f}{n})$-minimal-batch-order-fairness (Definition \ref{defn:minimalbatch}) is only meaningful if $f< \frac{n}{2}$.
}

In the rest of this work, we assume that every replica eventually votes on every transaction.
This may not necessarily hold for faulty replicas.  A precise choice of response to this type of faulty behavior
will depend on the network model and communication protocol, but, informally,
if a replica fails to vote on a transaction within a ``sufficient'' time,
it suffices for the honest replicas to fabricate a vote on behalf of the faulty replica.
\S \ref{sec:instantiation} gives an example construction that satisfies this property,
but note that prior work on this problem (e.g., \cite{kelkar2020order,themis,cachin2022quick}) implicitly
use the same ideas to handle unresponsive replicas.

\section{Defining Fair Ordering}
\label{sec:orderdefn}

\subsection{Batch Order Fairness}

Our discussion here follows the notion of ``fairness'' discussed in prior work (e.g., \cite{kelkar2020order,themis,cachin2022quick,kiayias2024ordering}).
The natural notion of ``fairness'' in prior work, which we follow here,
is that if ``most'' replicas report a transaction $\tx$ before another transaction $\txp$, then the protocol outputs $\tx$ before $\txp$.
Making this notion precise requires accounting for what are known as Condorcet cycles.
A (super)majority might report $\tx$ before $\txp$, another majority reports $\txp$ before $\txpp$, and a third majority reports
$\txpp$ before $\tx$.

Aequitas \cite{kelkar2020order} sidesteps this problem by outputting these Condorcet cycles in ``batches,'' which 
leads to the following definition (paraphrased from \cite{kelkar2020order}).

\iffalse
The following informal definition (paraphrasing from Definition III.1 of \citet{themis}) attempts to capture the notion
Honest replicas might receive transactions over the network in a different order.
%Honest replicas submit as their ordering vote the order in which they receive transactions over the network,
%but might vote differently due to e.g. variability in network delay.
To resolve disagreements ``fairly,'' recent work \cite{kelkar2020order}
proposes a notion known as ``$\gamma$-Batch-Order-Fairness,'' which we reproduce below from Definition III.1 of \citet{themis}
(a more complicated but functionally similar definition is used in \citet{kelkar2020order}).
\fi

\begin{definition}[$\gamma$-batch-order-fairness]
\label{defn:gammabatch}
Suppose that $\tx$ and $\txp$ are received by all nodes. 
If $\gamma n$ nodes received $\tx$ before $\txp$, then a ranking algorithm never outputs $\txp$ in a later batch than $\tx$.
\end{definition}

In other words, $\tx$ and $\txp$ could be in the same batch.  Transactions within a batch are ordered \textit{arbitrarily}.
Subsequent work \cite{themis,cachin2022quick,kiayias2024ordering} follows this same definition (or a restatement of it).

There are two key problems with this definition.
First, this definition does not preclude a ranking algorithm from putting all transactions into one batch
(and therefore ordering transactions arbitrarily).  For this definition to be meaningful,
it needs a notion of ``minimality'' to batch sizes.
We address this problem by strengthening the definition to include a notion of ``minimality'' (Definition \ref{defn:minimalbatch}).
Some of the prior work \cite{kelkar2020order,cachin2022quick,kiayias2024ordering}, but not all \cite{themis}, implicitly satisfies this definition.

Second, it is not clear whether a higher or lower $\gamma$ gives stronger fairness properties, yet 
all of the prior work requires a fixed parameter choice.
Indeed, there are simple examples (with no faulty replicas) 
where Definition \ref{defn:minimalbatch} only implies ordering restrictions
at high values of $\gamma$, others at only low values, and still others at only intermediate values (\S \ref{apx:orderexamples}).
It is not the case that satisfying Definition \ref{defn:minimalbatch} for a low value of $\gamma$ implies 
satisfying the definition for a high value of $\gamma$.
We address this problem by constructing a protocol that satisfies Definition \ref{defn:minimalbatch} 
for all $\frac{1}{2}\leq \gamma\leq 1$ simultaneously.

%Prior work relies on a fixed choice of $\gamma$, but
%our algorithms satisfy Definition \ref{defn:minimalbatch}
%for all values of $\gamma$ simultaneously.

\iffalse
Aequitas \cite{kelkar2020order} outputs transactions in batches, and ``no later than'' means ``in the same batch.''%for a protocol that outputs transactions in batches.  
Transactions within a batch are ordered arbitrarily.  Themis \cite{themis} outputs a total ordering, with the assertion
that this total ordering could be divided into contiguous (disjoint) batches satisfying this property.
\fi

\subsection{Batch Minimality}
\label{sec:minimality}

\iffalse
Unfortunately, Definition \ref{defn:gammabatch} is insufficient to provide meaningful ordering guarantees.
A protocol that outputs all
transactions in the same batch vacuously satisfies this definition, and if transactions within a batch can be ordered arbitrarily,
then every ordering satisfies Definition \ref{defn:gammabatch} for every $\gamma$.

Prior work \cite{themis} recognizes the need for ``minimality''.%, and therefore argues that output batches are ``minimal'' in size.
However, the notion of minimality in prior work is not achievable in the presence of faulty replicas (Theorem \ref{thm:batch_impossible}).
Towards a feasible notion of minimality, we first define a graph denoting ordering dependencies.
\fi

To construct a precise definition, we start with a weighted dependency graph
over a set of transactions.

\begin{definition}[Ordering Graph]
\label{defn:orderinggraph}
Suppose that replicas report ordering votes $(\sigma_1,...,\sigma_n)$ on a set of transactions $V$.
The \textit{Ordering Graph} $G(\sigma_1,...\sigma_n)$ is a complete, weighted directed graph with vertex set $V$
and, for each transaction pair $(\tx, \txp)$,
an edge of weight $w(\tx, \txp)=\alpha$ if $\alpha n$ replicas report $\tx\before\txp$.
\end{definition}

%The ordering graph contains edges in both directions between $\tx$ and $\txp$.  
%Specifically, it contains an
%edge $(\tx, \txp)$ of weight $\alpha$, and an edge $(\txp, \tx)$ of weight $(1-\alpha)$
%(so $w(\tx, \txp) + w(\txp, \tx)=1$ for all $\tx,\txp$).

Any notion of ``minimal'' should at least require that when there are no cycles in ordering dependencies, the output is
consistent with every ordering dependency
(as in Theorem IV.1 part 1, Themis \cite{themis}).  Exact $\gamma$-minimality generalizes this notion to the case where cycles in ordering
dependencies exist.

\begin{definition}[Exact $\gamma$-minimality]
\label{defn:strongminimal}
For any pair of transactions $\tx, \txp$, if $\gamma n$ replicas receive $\tx \before \txp$ but $\txp$ is output before $\tx$,
then there exists a sequence of transactions $\txp=\tx_1,...,\tx_k=\tx$ where
$\gamma n$ replicas receive $\tx_i\before \tx_{i+1}$ and $\tx_i$ is output before $\tx_{i+1}$ for all $i$.
\end{definition}

However, this notion of minimality is \textit{impossible} to achieve (for any $\gamma$) in the face of \textit{any} faulty replicas.
The condition on $f$ in Lemma \ref{lem:batch_impossible} is a technical artifact 
of a construction in the proof \aaaielide{(related to the pigeonhole principle)}{}. 
For some parameter settings, such as $n=7$ with $\gamma=\frac{2}{3}$,
the condition admits $f=1$.

\begin{lemmarep}
\label{lem:batch_impossible}
No protocol can achieve $\gamma$-batch-order-fairness with exactly $\gamma$-minimal output batches for any $n$ and $f$ greater than
$n\mod (n-\lceil \gamma n\rceil + 1)$
(for $\gamma > 1/2$).
\end{lemmarep}

In fact, the construction in the proof contains no cycles in the ordering dependency graph of weight $\gamma$,
thereby applying even to a weaker version of Definition \ref{defn:strongminimal} that applies only to the case where the ordering graph,
restricted to edges of weight at least $\gamma$, has no cycles.
%\footnote{The seemingly contradictory nature of this statement, as compared to prior work,
%relates to a subtle flaw in the definition of an ordering dependency in prior work, which we discuss in \S \ref{sec:priorcomparison}.}

\begin{proof}
Define m=$\lceil \gamma n\rceil$, and $k=\lfloor \frac{n}{n-m+1}\rfloor$.
Consider the process of choosing an ordering between $k$ transactions $\tx_1,\tx_2,...,\tx_k$.
Assume that every node submits a vote on transaction ordering (i.e. the theorem
holds even when faulty nodes are required to submit a vote).

An ordering dependency arises, therefore, if and only if there are at least $m$ votes for some $\tx$ before another $\txp$.

Divide the set of nodes into disjoint groups $G_i$ for $1\leq i \leq k$ of size $n-m+1$.
The remainder are the faulty nodes.

Suppose that nodes in group $i$ receive transactions in order $\tx_i \before ... \before \tx_k \before \tx_1 \before ... \before \tx_{i-1}$,
and that the faulty nodes receive the same ordering of transactions as the nodes in group $G_1$.

As such, the transaction pairs $\tx_i, \tx_j$ that receive at least $m$ votes are exactly those with $i<j$,
and there are no cycles in this ordering dependency graph.  Thus, any protocol must output the ordering
$\tx_1,...,\tx_k$.

However, for any protocol that lacks knowledge of which nodes are faulty,
this scenario is indistinguishable from one where
the faulty nodes had received transactions in the ordering received by group $G_2$,
but reported the ordering observed in $G_1$.  As such, the protocol
would necessarily output the ordering observed by $G_1$.

However, in this case, the only correct output would have been $\tx_2,...,\tx_k, \tx_1$,
as only the transaction pairs $\tx_i, \tx_j$ with $i<j$ for $i\neq 1$ or $j=1$
received $m$ votes for $\tx_i \before \tx_j$. \safeqed{}
%Note that in this second scenario, there are no cycles in the ordering dependencies (as defined in Definition \ref{defn:gammabatch}),
%but the output cannot satisfy the minimality condition of Theorem IV.1 part 1, \cite{themis}. \safeqed
\end{proof}
%7 - ceil (7 * 2/3) + 1 = 
%7 -5 + 1 = 3

%m= ceil(gamma n)
%Groups of size $(n-m+1)$
%If gamma = a/b, then need a divides n/b
%so let n = lcm(a,b)

The core difficulty is that faulty replicas can cause the weight on an edge observed by any algorithm to be different from the true,
ground truth weight. 
Intuitively,
if a protocol has a reliable communication layer (i.e. every honest replica is able to submit a vote on its ordering preference),
then the worst that a faulty replica can do is to misreport its ordering preferences.  If there are $f$ faulty replicas, then
the faulty nodes can artificially reduce or increase the fraction of replicas that vote for $\tx \before \txp$ by at most $\frac{f}{n}$.
As a protocol cannot distinguish a faulty replica from an honest one by its votes,
we therefore must take some error in edge weights into account in our definition of minimality. 

\begin{definition}[$(\gamma,\delta)$-minimal-batch-order-fairness]
\label{defn:minimalbatch}
An ordering is $(\gamma, \delta)$-minimally-batch-order-fair if, for any transaction pair $(\tx, \txp)$
that is received in that order by at least $\gamma n$ replicas but output by the protocol in the reverse ordering,
then there is a sequence of transactions $\txp = \tx_1,...,\tx_k=\tx$ where at least $(\gamma-2\delta)n$ replicas
receive $\tx_i\before \tx_{i+1}$ and $\tx_i$ is output before $\tx_{i+1}$.
%then there is a path in the ordering graph from $\txp$ to $\tx$ where each edge has weight at least $\gamma-2\delta$.
\end{definition}

Definition \ref{defn:minimalbatch} captures the notion that a protocol cannot distinguish between
a $\delta$-fraction of replicas misreporting a transaction ordering.  Given this indistinguishable $\delta$ fraction,
the protocol outputs minimally-sized batches.  Definition \ref{defn:strongminimal} corresponds to the case of $\delta=0$.

This definition does not explicitly discuss ``batches'' or ``minimality,'' but approximately minimal batches (the strongly connected components of Lemma \ref{lemma:batchlang}) can be recovered 
from any ordering satisfying it.  The second condition of Lemma \ref{lemma:batchlang} limits the size of output batches.

\begin{lemmarep}
\label{lemma:batchlang}
Suppose that an output ordering satisfies $(\gamma, \delta)$-minimal-batch-order-fairness.
Compute the ordering graph $\hat{G}$, drop all edges with weight below $\gamma-\delta$, and compute the strongly connected components 
of the remainder.
\begin{itemize}
    \item If $\gamma n$ replicas received $\tx \before \txp$, then either $\tx$ and $\txp$ are in the same strongly connected component,
    or all transactions in the component containing $\tx$ are output before any transactions in the component containing $\txp$.

    \item If $\gamma n$ replicas receive $\tx \before \txp$ and there is no sequence of transactions 
    $\txp = \tx_1,...,\tx_k=\tx$ where at least $(\gamma-2\delta)n$ replicas
    receive $\tx_i\before \tx_{i+1}$, then all transactions in the component containing $\tx$ are output before any transactions in the component containing $\txp$.
\end{itemize}
\end{lemmarep}

\begin{proof}
If at least $\gamma n$ replicas receive $\tx\before \txp$, then the edge $(\tx, \txp)$ is included in the thresholded ordering graph,
and if less than $(\gamma - 2\delta)$ replicas receive $\tx \before \txp$, then the edge is not included.

Thus, if at least $\gamma n$ replicas receive $\tx\before\txp$, then either $\tx$ and $\txp$ are in the same strongly connected component,
or there is an edge from the component containing $\tx$ to that containing $\txp$.  And if, additionally, there is no sequence of transactions
from $\txp$ to $\tx$ as in the lemma statement, then $\tx$ and $\txp$ must be in different components.

If there is an edge from one strongly connected component to another, then all transactions in the first component must be output
before any in the second (or else a violation of Definition \ref{defn:minimalbatch}) would occur.
\safeqed{}
\end{proof}

There is one subtle difference between this definition and 
%(an analogue of) the notion
%of ``minimality'' in Theorem IV.1 of Themis, and 
the requirement for an explicit sequence of output batches in prior work (e.g., in \cite{kelkar2020order,themis}).
When there are two disjoint strongly connected components with no dependencies (of strength at least $\gamma$)
from one to the other, Definition \ref{defn:minimalbatch} allows the output to interleave these components.
This may be strictly required (\S \ref{apx:interleave}).

\subsection{Comparison to Prior Work}
\label{sec:priorcomparison}

Before constructing our protocol, we first apply Definition \ref{defn:minimalbatch} to related prior work.  Aequitas \cite{kelkar2020order}
is instantiated with a choice of parameter $\gamma$ and a bound $f$ on the number of faulty replicas. 

\begin{lemmarep}
\label{lemma:aeq_batch}
Aequitas \cite{kelkar2020order} with parameters $(\gamma, f)$ achieves $(\gamma, \frac{f}{n})$-minimal-batch-order-fairness.
\end{lemmarep}

Note that Aequitas requires $\gamma - \frac{1}{2} > \frac{2f}{n}$.  Lemma \ref{lemma:aeq_batch} is a stronger version of
 Theorem 6.6 of \cite{kelkar2020order}.

\begin{proof}
Aequitas considers a transaction pair $(\tx, \txp)$ as an ordering dependency if at least $\gamma n - f$ report
receiving $\tx$ before $\txp$ (part 3.b.ii of \S 6.2 of \cite{kelkar2020order}).  
This ensures that the protocol considers every ordering dependency
that truly receives at least $\gamma n$ votes, and might include any edge that receives at least $\gamma n - 2f$ votes.

Aequitas sequentially outputs minimal cycles in its computed graph of ordering dependencies (not Definition \ref{defn:orderinggraph}),
ordering within each cycle arbitrarily.
This means that if some pair that receives at least $\gamma n$ votes is output in reversed order,
then it must be part of a cycle of dependencies, and each of these dependencies must have received at least $\gamma n - 2f$ votes in order
to be included in Aequitas's computed graph. \safeqed
\end{proof}

Broadly speaking, Aequitas, along with the protocols of \citet{cachin2022quick} and \citet{kiayias2024ordering},
follow the same pattern.  Through (very different) communication protocols,
they come to agreement on, for each transaction pair $\tx$ and $\txp$, the fraction of replicas that report $\tx$ before $\txp$.
These protocols choose (in advance) a fixed parameter $\gamma$, and then drop from consideration ordering dependences of strength less
than $\gamma$.  This is equivalent to constructing the ordering dependency graph of Definition \ref{defn:orderinggraph},
but then dropping all edges of weight less than $\gamma$ and then forgetting all of the weights on the remaining edges.  The strongly connected
components of this graph form the batches output in Aequitas.

Aequitas (as well as the protocols of \citet{cachin2022quick} and \citet{kiayias2024ordering}) also is not asymptotically live 
(Definition \ref{defn:asymptotic_liveness}).  The strongly connected components of this unweighted dependency graph
can be arbitrarily large, and these protocols wait until they observe the entirety of a component before outputting any transaction
in it.  An exception is one version of the protocol in \citet{kiayias2024ordering}, which adds timestamps to transactions
and (given precise network assumptions) can output a component in a streaming fashion.

By contrast, Ranked Pairs (Theorem \ref{thm:rp_fair}) achieves $(\gamma, \frac{f}{n})$-minimal-batch-fairness for every $\gamma$ simultaneously,
and for any $f$ (our work does not require knowledge of a bound on $f$).  Our protocol gets these stronger results by using all of the information 
available in the problem input.  Additionally, Ranked Pairs guarantees a bounded liveness delay (although precise end-to-end results depend also on network conditions).

Themis \cite{themis} claims (Theorem IV.1) that its output can be partitioned into ``minimal'' batches,
which appears to contradict Lemma \ref{lem:batch_impossible}.
However, Themis uses a much weaker notion of ``minimal.''

Where $\gamma$-batch-order-fairness would create a dependency for $\tx$ not after $\txp$ if at least $\gamma n$ replicas
receive $\tx$ before $\txp$, 
the relaxation implied by Definition III.1 of \cite{themis} creates such a dependency if at least $n(1-\gamma) + 1$ receive $\tx$ before $\txp$
(that is, it cannot be the case that more than $\gamma n$ receive $\txp$ before $\tx$).
Theorem IV.1 of \cite{themis} considers its output batches ``minimal'' if they are singletons when there are no cycles of these dependencies.
%The minimality statement of Theorem IV.1 of \cite{themis} is with regard to this much weaker notion of dependency.
This notion of minimality is much weaker than what we discuss here, and does not provide any 
guarantee when there is significant disagreement between replicas on 
transaction orderings.  Lemma \ref{lemma:inverted} gives an example highlighting this distinction.

\begin{lemmarep}
\label{lemma:inverted}
In Themis \cite{themis}, even if all nodes report $\tx$ before $\txp$ and there are no faulty replicas
nor cycles in ordering dependencies (as defined in Definition \ref{defn:gammabatch}) for any $\gamma>1/2$,
 the output may put $\txp$ before $\tx$.
\end{lemmarep}

This kind of counterintuitive behavior is a possible output of Themis's protocol.  
When two transactions $\tx, \txp$ receive
$n-f$ votes in one round of Themis, it determines whether to consider as an ordering dependency $\tx \before \txp$ or $\txp \before \tx$ 
by \textit{majority vote}.  Equivalently, Themis as a protocol operates in the regime of $\gamma=\frac{1}{2}$, regardless
of a chosen parameter value 
(the choice of $\gamma$ does affect how Themis considers 
transactions reported by some but not all replicas, which does not affect this discussion).

\begin{proof}

Suppose that the number of nodes is even, and break the nodes into two groups.
Group one receives transactions $\tx_1 \before \tx_2 \before \tx_3$, while group two 
receives transactions $\tx_3 \before \tx_1 \before \tx_2$.

All nodes submit their local orderings to the protocol correctly.  Themis proceeds in rounds;
the round under consideration consists only of these three transactions.

Themis builds a graph of what it considers ordering dependencies on the transactions in a round
(Figure 1, part 1, \cite{themis}).
Let $w(\tx, \txp)$ be the number of replicas that vote
for $\tx$ before $\txp$.

An ordering dependency between $\tx$ and $\txp$ is included if (1) $w(\tx, \txp) \geq w(\txp, \tx)$
and (2) $w(\tx, \txp) \geq n(1-\gamma) + f + 1$.  Themis assumes $n>4f/(2\gamma -1)$, or equivalently,
$2f < n (\gamma - 1/2)$, so $(1-\gamma) + (f+1)/n \leq (1-\gamma) + (2f)/n \leq (1-\gamma) + (\gamma-1/2) \leq 1/2$.
Ties are broken in an unspecified (deterministic) manner.

As such, (depending on the tiebreaking), Themis may consider as valid ordering dependencies the pairs $(\tx_1, \tx_2)$,
$(\tx_2, \tx_3)$, and $(\tx_3, \tx_1)$.  Within a strongly connected component of its dependency graph, Themis computes a Hamiltonian cycle,
then outputs transactions by walking arbitrarily along that cycle.  As such, a possible output of Themis is the ordering
$\tx_2 \before \tx_3 \before \tx_1$.

Thus, all nodes received $\tx_1$ before $\tx_2$, and there were no cycles of ordering dependencies for any $\gamma >1/2$, but the protocol could
output $\tx_2$ before $\tx_1$.

Note that Themis only waits for $n-f$ votes from replicas before proceeding, not the full $n$.
This difference is immaterial if $f=0$.
However, the construction above works only if the numbers of votes for 
$\tx_3 \before \tx_1$ and for $\tx_2 \before \tx_3$ are both at least $n(1-\gamma) + (f+1)$.
Tighter analysis shows that 
$n(1-\gamma) + (f+1) = n(1-\gamma) + (2f) - (f-1) \leq n(1-\gamma) + n(\gamma - 1/2) - (f-1) = n/2 - (f - 1) < (n-f)/2$, so the construction still works.

For $\gamma$ larger than $1/2 + \varepsilon$ for some small $\varepsilon$ and sufficiently large $n$,
this construction could be repeated with more transactions in the cycle.
This eliminates the need to abuse a tiebreaking rule when votes are evenly divided. \safeqed
\end{proof}

\iffalse

Themis appears to choose an intermediate definition.
The protocol builds an ordering dependency graph (an unweighted analogue of Definition \ref{defn:orderinggraph}).  
For any two transactions $\tx$, $\txp$,
this graph (by the time sufficient nodes have voted on $\tx$ and $\txp$, so that the transactions are output)
will contain either a directed edge from $\tx$ to $\txp$ or vice versa.
In one round, if both transactions are voted on by $n-f$ replicas (i.e. most transaction pairs, under reasonable network assumptions)
then the choice is determined, in effect, \textit{solely} by majority vote
(thereby disregarding the parameter $\gamma$).  Ordering dependencies

In fact, Themis creates a dependency for exactly one of $\tx<\txp$ or $\txp < \tx$ by \textit{majority vote}
among replicas (effectively disregarding the parameter $\gamma$).

Unfortunately, Themis \cite{themis} fails to achieve any meaningful ordering guarantees, even when there are no faulty nodes.

The difficulty comes from the fact that by design, Themis builds a complete directed graph to determine an ordering.  This includes every transaction
pair $(\tx, \txp)$ where $\gamma n$ nodes vote for $\tx < \txp$, but also includes an edge if only $\frac{1}{2}+\varepsilon$ vote for $\tx< \txp$, thereby
introducing spurious ordering dependencies.

\fi

\section{Ranked Pairs Voting}
\label{sec:rpv}

Our protocol is based on the Ranked Pairs method \cite{tideman1987independence}.  
Ranked Pairs (Algorithm \ref{alg:rpv}, reproduced in \S \ref{apx:algorithms})
first builds the ordering graph (Definition \ref{defn:orderinggraph}),
then iterates through edges in order of edge weight (breaking ties arbitrarily).
It builds an acyclic set of edges greedily, adding each new edge so long as it does
not create a cycle with edges already in the set.  The output ordering
is the result of topologically sorting the resulting acyclic graph.

\begin{theoremrep}
\label{thm:rp_fair}
Given a ordering vote (on every transaction in a finite set) from every replica,
Ranked Pairs Voting simultaneously satisfies $(\gamma, \frac{f}{n})$-minimal-batch-order-fairness
for every $\gamma$, and does not depend on any fixed bound on $f$.
\end{theoremrep}

Lemma \ref{lemma:batchlang} and Theorem \ref{thm:rp_fair} together imply an important property.
Ranked Pair's output
can be divided into batches consistent with $\gamma$-batch-order-fairness for \textit{any} $\gamma$
(subject to the interleaving discussed in \S \ref{sec:minimality}).
There is no need to choose a $\gamma$ or arbitrarily order transactions within a batch.

\begin{proof}

Ranked Pairs Voting ensures that if some edge of weight $\gamma$ is not included in the output graph $G$ --- that is to say,
if $\gamma n$ nodes vote for $\tx$ before $\txp$, but the output ordering has $\txp \before \tx$ --- then there must be a directed path of edges
in $G$ already from $\txp$ to $\tx$.  As the algorithm looked at these edges before the current edge,
these edges must have weight (as observed by the algorithm) at least $\gamma$.

Note that if there are $f$ faulty nodes, these nodes can, at most, adjust the weight (observed by the algorithm) on any particular transaction pair
by at most $\frac{f}{n}$.

Thus, an edge with true weight $\gamma$ is only reversed if there exists a path in the ordering graph
(that is included in $G$)
in the opposite direction
of minimum true weight at least $\gamma - \frac{2f}{n}$.

Nowhere in the algorithm do its choices depend on $\gamma$ or $f$ (or $n$). \safeqed
\end{proof}

%For ease of exposition, we will occasionally write of the decision that Ranked Pairs makes when applied to countably infinite sets of transactions.
%In this case, we will say that Ranked Pairs chooses a subset of edges such that an edge is present in the subset if and only if
%it can be added to the (possibly infinite) set of edges earlier than it in the ordering (of the edges) without creating a cycle (such a set exists, via Zorn's Lemma, and is unique).

\section{Streaming Ranked Pairs}
\label{sec:streaming_rpv}

We now turn to the streaming setting that is the focus of our work.
There may be an infinite set of transactions to consider overall, but at any finite time, 
an algorithm must compute an ordering on a subset of the transactions that have been seen thus far.
In a real-world system, such an algorithm would be run with some periodic frequency,
and replicas would periodically append to their ordering votes.

An algorithm must be \textit{monotonic} (Definition \ref{defn:mono})---%
if replicas extend their ordering votes, the algorithm, when run again on the extensions,
can only extend its prior output.

\begin{definition}[Monotonicity]
\label{defn:mono}
A sequencing algorithm $\mathcal{A}(\cdot,\ldots,\cdot)$ is \textit{monotonic}
if, given two sets of ordering votes $\left(\sigma_1,\ldots,\sigma_n\right)$
and $\left(\sigma^\prime_1,\ldots,\sigma^\prime_n\right)$,
such that $\sigma^\prime_i$ extends $\sigma_i$ for all $i\in [n]$,
$\mathcal{A}(\sigma^\prime_1,\ldots,\sigma^\prime_n)$ extends $\mathcal{A}(\sigma_1,\ldots,\sigma_n)$.
\end{definition}

If a sequencing algorithm is monotonic,
it implies a well-formed definition for aggregating 
a set of orderings on countably infinite sets of transactions,
not just on finite sets.  $\tx$ comes before $\txp$
in the infinite case
if there exists a finite subset of the input orderings
such that the algorithm puts $\tx\before\txp$ in its output.

A non-vacuous condition is also required---the algorithm
that never outputs anything is monotonic.

\begin{definition}[Asymptotic Liveness]
\label{defn:asymptotic_liveness}
A sequencing algorithm $\mathcal{A}(\cdot,\ldots,\cdot)$ is \textit{asymptotically live}
if, given any set of countably infinite ordering votes $(\hat{\sigma}_1,\dots,\hat{\sigma_n})$
and any transaction $\tx$ in those votes,
there exists an $N$ such that when each $\hat{\sigma_i}$ is trunctated to the first $N$
elements of the ordering to produce $\sigma_i$,
$\tx$ is included in $\mathcal{A}(\sigma_1,\ldots,\sigma_n)$.
\end{definition}

For expository simplicity, we present
our algorithm in two steps.  First, we give a streaming algorithm
(Algorithm \ref{alg:stream_rpv}) that is monotonic but not asymptotically live.
Then, we give a modification that ensures asymptotic liveness (Algorithm \ref{alg:stream_rpv_2}).

%Corollary \ref{cor:stream_mono} shows that our version of Ranked Pairs,
%Algorithm \ref{alg:stream_rpv}, is monotonic.
%However, Algorithm \ref{alg:stream_rpv} is not asymptotically live (\S \ref{apx:non_live}).
%Theorem \ref{thm:asymptotic_liveness} shows that a modification to the tiebreaking rule
%between edges of equal weight, as in
%Algorithm \ref{alg:stream_rpv_2}, produces an asymptotically live algorithm.

%In a deployment of a sequencing algorithm in a decentralized system,
%replicas would periodically append to their ordering votes (recorded e.g. some append-only log) 
%as they receive new transactions.
%A sequencing algorithm would be run with some temporal frequency on these votes.

\subsection{Decisions Are Local}

Implementing Ranked Pairs voting in this streaming setting requires first 
understanding when the algorithm chooses or rejects an edge in the ordering graph.

%For expository simplicity, we occasionally write
%of the decision that Ranked Pairs would make on some edge $(\tx, \txp)$
% if given a complete set of votes from every replica, for all (possibly infinitely many)
%transactions.  To be precise, this language means the decision that Ranked Pairs would make 
%on the edge if given the set of votes on all transactions in any finite set $V\supseteq (P_{\tx}\cup Q_{\tx}$.
%Lemma \ref{lemma:local} makes this language well-defined. 
%Note that $P_{\tx}$ and $Q_{\tx}$ are finite, and all replicas must vote on all transactions in these sets
%before some finite time. 

\begin{lemmarep}
\label{lemma:cycles}
If Ranked Pairs rejects an edge $(\tx, \txp)$ of weight $\alpha$, then it must choose some directed path of edges from $\txp$ to $\tx$, where
each edge has weight at least $\alpha$.
\end{lemmarep}

\begin{proof}
If Ranked Pairs rejects an edge $(\tx, \txp)$ of weight $\alpha$, then adding the edge would have created a cycle among the edges the algorithm had already chosen.  Ranked Pairs looks at edges in order of weight, so all of the edges on the already chosen path from $\txp$ to $\tx$ must have weight
at least $\alpha$.
\end{proof}

\begin{observationrep}
\label{obs:maxweight}
Ranked Pairs chooses every edge of weight $1$.
\end{observationrep}

\begin{proof}
There cannot be any cycles of edges, all of weight $1$.  Such a situation would require every replica to submit a cycle of preferences,
which is impossible, given that each ordering vote $\sigma$ is a total, linear order. \safeqed
\end{proof}

%For simplicity of exposition, we occasionally write of the decision that Ranked Pairs would make on some edge $e=(\tx, \txp)$ if
%given a complete set of votes from every replica, for all transactions (which may be a countably infinite set).

%The set of edges visited by Ranked Pairs prior to $e$ must be finite,
%so there exists a set of transactions $V_e$ containing the endpoints of all of these edges.
%Ranked Pairs, run only on the ordering graph restricted to $V_e$, must make the same choice about $e$
%as it would when run on any set $V^\prime \supset V_e$.  To be precise, therefore, discussion of 
%the decision of Ranked Pairs on $e$, when run on 
%a countably infinite set of transactions, refers to the decision of Ranked Pairs when run on any sufficiently large finite subset.

%Lemma \ref{lemma:cycles} implies that a
%a streaming version of ranked pairs can make a decision on an edge $(\tx_A, \tx_B)$ once the algorithm knows the algorithm's decision
%on every possible edge of every possible path from $\tx_B$ to $\tx_A$ that the non-streaming version of ranked pairs might choose.

Furthermore, the set of vertices that such a path can visit are bounded to a local neighborhood of $\tx_A$ and $\tx_B$.
%For any transaction $\tx$, divide the set of all other transactions into the following subsets.
\begin{definition}
Let $\tx$ be any transaction.
\begin{enumerate}
  \item Let $P_{\tx}$ be the set of all \textit{preceeding} transactions; that is, all $\txp$ with $w(\txp, \tx)=1$.
  \item Let $Q_{\tx}$ be the set of all \textit{concurrent} transactions; that is, all $\txp$ with $0<w(\txp, \tx)<1$.
  \item Let $R_{\tx}$ be the set of all \textit{subsequent} transactions; that is, all $\txp$ with $w(\txp, \tx)=0$.
\end{enumerate}
\end{definition}

%Naturally, to avoid creating a cycle, ranked pairs cannot choose path from $\tx_B$ to $\tx_A$ that visits
%an edge in $R_{\tx_A}$.

\begin{lemmarep}
\label{lemma:local}
If ranked pairs chooses all edges in a path from $\tx_B$ to $\tx_A$,
no transactions on the path are in $R_{\tx_A}$ or $P_{\tx_B}$.
\end{lemmarep}

\begin{proof}
If the path visits a transaction $\tx$ in $R_{\tx_A}$, then it creates a cycle with the edge from $\tx_A$ to $\tx$
(which ranked pairs must choose, by Observation \ref{obs:maxweight}).
If the path visits a transaction $\tx$ in $P_{\tx_B}$, then it creates a cycle with the edge from $\tx$ to $\tx_B$.
\safeqed
\end{proof}

Lemma \ref{lemma:local} lets the streaming algorithm decide whether to choose an edge given only a finite amount of
information.
% seeing only a small amount of additional
%information from replica votes.  
Once the algorithm sees a vote for a transaction $\tx$ from every replica, it can compute the weight in the ordering graph
for every edge $(\tx, \txp)$ and $(\txp, \tx)$.  
Unseen transactions must be in $R_{\tx}$.%, which the algorithm need not directly consider.

Even so, a streaming algorithm cannot always know whether or not Ranked Pairs would choose an edge.  
As such, we allow our algorithms to leave an edge in an ``\textit{indeterminate}'' state, and design
our algorithms to account for indeterminate edges when considering subsequent edges.

First, we construct an ordering graph that captures the information available at a finite time.

\begin{definition}[Streamed Ordering Graph]~\\
\label{defn:streamedgraph}
Suppose that each replica submits an ordering vote $(\sigma_1,...,\sigma_n)$.
Let $V$ be the set of all transactions that appear in each vote, and let $\hat{v}$ be a new vertex (the ``future'')
The \textit{Streamed Ordering Graph} is a weighted, directed, complete graph $\hat{G}=(V^\prime,E)$
on vertex set $V\cup \lbrace \hat{v}\rbrace$.

For each $\tx$ and $\txp$, set weights $w(\tx, \txp)$ and $w(\txp, \tx)$ as in Definition \ref{defn:orderinggraph}.
Set $w(\tx,\hat{v})=1$ for all $\tx$.  
If there exists $\txp$ that appears in the votes of some
but not all replicas and which preceeds $\tx$ in at least one replica's vote, set $w(\hat{v}, \tx)=1$, and otherwise $w(\hat{v}, \tx)=0$.
\end{definition}

Conceptually, the ``future'' vertex $\hat{v}$ represents all transactions that have not received votes from all replicas.
Implicitly, all edge weights not computable from the available information are upper-bounded by $1$.
%between these transactions, and between these transactions and the other transactions already voted on by all replicas,
%is upper bounded in this graph.
\aaaielide{%
An implementation might compute a tighter upper bound on the weights of edges $(\hat{v}, \tx)$.  Subsequent arguments require only that
the assigned weight is an upper bound of the true weight.}{}

\subsection{A Streaming Algorithm}
\label{sec:streaming_alg}

We can now construct a streaming algorithm to compute Ranked Pairs.
%Assume that edges with equal weight are sorted according to the same tiebreaking rule used in Algorithm \ref{alg:rpv}
%(as edges adjacent to $\hat{v}$ start marked as \textit{indeterminate}, the tiebreaking rule
%need not be extended to compare these edges with edges between real transactions).
%\S \ref{sec:liveness}
%will manipulate the tiebreaking rule to reduce the delay caused by the ordering protocol.
Instead of either including or rejecting every edge,
Algorithm \ref{alg:stream_rpv} (shown in full in Appendix \ref{apx:algorithms}) has the option to declare an edge as \textit{indeterminate};
that is,
the algorithm cannot determine whether to include or exclude the edge
with the information in the input.

The algorithm proceeds identically to the classic Ranked Pairs algorithm, except regarding how it treats \textit{indeterminate} edges.
For each new edge, it first checks whether it can accept the edge (without creating a cycle of accepted edges),
supposing that all currently \textit{indeterminate} edges are later accepted.
Conversely, it also checks whether it can reject the edge, 
which means checking whether accepting the edge would always create a cycle of accepted edges
(assuming the currently \textit{indeterminate} edges are later rejected).
The algorithm operates on the modified input graph (Definition \ref{defn:streamedgraph}), where the ``future'' vertex
bounds the information the algorithm has to consider about as-yet-unseen transactions.  Informally,
all edges between as-yet-unseen transactions are \textit{indeterminate}, so from the perspective of directed connectivity analysis,
the unseen transactions form a single, strongly-connected component.

To prove correctness of this algorithm, it suffices to show that the output of the algorithm is
monotonic %---additional votes from replicas can only cause transactions to be appended to the output---
and
that the output is consistent with the non-streaming version of Ranked Pairs.
By consistency, we specifically mean that at any point, if clients were to stop sending new transactions and all replicas eventually voted
on every replica, Algorithms \ref{alg:rpv} and \ref{alg:stream_rpv} would produce the same output.

\begin{lemmarep}
\label{lemma:same_choices}
Consider a counterfactual scenario where clients stop sending transactions, all replicas eventually receive every transaction,
and every replica includes every transaction in its output vote.

Whenever Algorithm \ref{alg:stream_rpv} includes a \textit{determinate} edge in $H$ (resp. excludes an edge),
that edge is included (resp. excluded) in the output of the non-streaming Ranked Pairs on the counterfactual input.
\end{lemmarep}

\begin{proof}
Suppose that the lemma statement holds for all edges earlier in the Ranked Pairs ordering.
Specifically, assume that Ranked Pairs would choose every higher \textit{determinate} edge,
not choose any higher rejected edge, and might or might not choose any higher \textit{indeterminate} edge (and that all edges of higher weight
are all either already rejected, chosen as \textit{determinate}, or chosen as \textit{indeterminate}).

By Lemma \ref{lemma:local},
Ranked Pairs would choose $(\tx_i, \tx_j)$ if and only if there does not exist a path of already chosen edges higher in the ordering
from $\tx_j$ to $\tx_i$ that does not enter $R_{\tx_i}$ or in $P_{\tx_j}$.

If there exists a path of \textit{determinate} edges in $U_{\tx_i,\tx_j}$,
then Streaming Ranked Pairs rejects $(\tx_i, \tx_j)$, and if there does not exist
a path of \textit{determinate} or \textit{indeterminate} edges, 
then Streaming Ranked Pairs accepts $(\tx_i, \tx_j)$.  Otherwise, the edge is left \textit{indeterminate}.

Note that Streaming Ranked Pairs considers edges in the same ordering that Ranked Pairs would (relative to the restricted set
considered).  However, at any point, due to the initialization of $H$ with \textit{indeterminate} edges, 
the set of edges in which the algorithm searches for a cycle might include more edges than those that would be considered
by Ranked Pairs.  However, these extra edges are \textit{indeterminate} and can therefore only make Streaming
Ranked Pairs choose \textit{indeterminate} for a considered edge, satisfying the induction hypothesis.  

Importantly,
the weights of these extra edges upper bound their true (unknown) weights.
And any cycle that would go through (an) as-yet-unseen transaction(s) maps to one (using \textit{indeterminate} edges)
through the future vertex $v$.  As such, 
the set of paths considered in Streaming Ranked Pairs is always a superset of that considered by Ranked Pairs,
and the difference between these sets is always made up of \textit{indeterminate} edges.

As such, when the streaming algorithm does not choose to leave an edge as \textit{indeterminate},
the streaming algorithm makes the same decisions as the non-streaming algorithm.
Thus, the induction hypothesis holds for $(\tx_i, \tx_j)$.

The induction hypothesis clearly holds for the first edge considered, so the lemma holds.
\safeqed
\end{proof}

The same argument shows that Algorithm \ref{alg:stream_rpv} is monotonic.

\begin{corollaryrep}
\label{cor:stream_mono}
Algorithm \ref{alg:stream_rpv} is monotonic.
\end{corollaryrep}

\begin{proof}

The proof of Lemma \ref{lemma:same_choices} actually shows that an edge is marked \textit{determinate} or rejected only if Ranked Pairs
would make that decision on the edge on \textit{any} counterfactual scenario that extends the input to the algorithm.

\safeqed
\end{proof}

Lemma \ref{lemma:same_choices} directly implies that the output of Algorithm \ref{alg:stream_rpv} matches that of Ranked Pairs.

\begin{theoremrep}
\label{thm:output_correct}
The output of Algorithm \ref{alg:stream_rpv} matches an initial segment of the ordering output by the non-streaming Ranked Pairs
(on the counterfactual of Lemma \ref{lemma:same_choices}).
\end{theoremrep}

\begin{proof}
Note that every transaction bordering on an indeterminate edge must come strictly after (in the topological ordering)
every transaction in the output of the algorithm.

By Lemma \ref{lemma:same_choices}, then regardless 
of whether or not the indeterminate edges are chosen in Ranked Pairs, 
every transaction bordering an indeterminate edge
must come after (in the true output of Ranked Pairs) all of the transactions output by the algorithm.

Within the set of output transactions, again because of Lemma \ref{lemma:same_choices}, transactions must be ordered
according to the true output of Ranked Pairs. \safeqed
\end{proof}

Ranked Pairs has the property that, for any set of inputs $(\sigma_1,\ldots,\sigma_n)$, if 
it is run on a restriction of the inputs to the transactions
that appear as an initial segment of its output on $(\sigma_1,\ldots,\sigma_n)$,
its output is unmodified on the restricted inputs (for completeness, we give a proof in \S \ref{apx:truncation_proof}).
As Algorithm \ref{alg:stream_rpv} is monotonic, 
its output exactly matches the Ranked Pairs ordering
when the input votes are restricted in this way.

\section{Liveness and Efficiency}
\label{sec:ordering:implementation}

\begin{algorithm}
\caption{Streaming Ranked Pairs}
\label{alg:stream_rpv_2}
\KwIn{An ordering vote from each replica $\lbrace \sigma_i \rbrace_{i\in [n]}$}
\KwIn{$\hat{H}$, the graph computed in the preceeding invocation of Algorithm \ref{alg:stream_rpv_2}}
%$\hat{G}=(V,E,w) \gets \hat{G}(\sigma_1,...,\sigma_n)$ \aaaielide{\tcc*[r]{Compute streamed ordering graph}}{\\}
%$H \gets (V,\emptyset)$ \aaaielide{\tcc*[r]{Initialize empty graph on same vertex set}}{\\}
\aaaielide{
$\hat{G}=(V,E,w) \gets \hat{G}(\sigma_1,...,\sigma_n)$ \tcc*[r]{Compute streamed ordering graph}
$H \gets (V,\emptyset)$ \tcc*[r]{Initialize empty graph on same vertex set}}
{
  $\hat{G}=(V,E,w) \gets \hat{G}(\sigma_1,...,\sigma_n)$, $H\gets (V, \emptyset)$\\
}
\ForEach{edge $(\tx, \hat{v})$ and $(\hat{v}, \tx)$ with $w(e)=1$}{
  Add $e$ to $H$, and mark it as \textit{indeterminate}
}
Add all \textit{determinate} edges of $\hat{H}$ to $H$, marked \textit{determinate}\\
%\ForEach{\textit{determinate} edge in $\hat{H}$}{
%  Add $e$ to $H$, and mark it as \textit{determinate}
%}
\ForEach{$\gamma$ in $(1,\frac{n-1}{n},\frac{n-2}{n},\ldots,0)$}
{
  $E_\gamma \gets \lbrace e\in E ~\vert w(e)=\gamma\rbrace$, ordered arbitrarily\\
 % Order $E_\gamma$ arbitrarily\\
  \Repeat{$E_\gamma$ is empty or no progress is made in one full loop over $E_\gamma$}{
    $e=(\tx_i, \tx_j) \gets$ first element of $E_\gamma$\\
    $U_{\tx_i, \tx_j} \gets (V\setminus (R_{\tx_i}\cup P_{\tx_j}))\cup \lbrace \hat{v}\rbrace$\\
    If there is a directed path in $H$ restricted to $U_{\tx_i,\tx_j}$
        from $\tx_j$ to $\tx_i$
        where every edge is \textit{determinate}, then do not include the edge in $H$.  Remove the edge from $E_\gamma$.\\
    Else if there is no directed path in $H$ restricted to $U_{\tx_i,\tx_j}$
        from $\tx_j$ to $\tx_i$
        where every edge is either \textit{determinate} or \textit{indeterminate},
        add $(\tx_i, \tx_j)$ to $H$ and mark it as \textit{determinate}.  Remove the edge from $E_\gamma$. \\
    Otherwise, defer $(\tx_i, \tx_j)$ to the end of the ordering on $E_\gamma$.
  }
  Mark all remaining edges in $E_\gamma$ as \textit{indeterminate}
}
\KwOut{The topological sort of $H$ (with sorting comparisons restricted to \textit{determinate} edges),
up to but not including the first transaction bordering an \textit{indeterminate} edge}
\end{algorithm}

\subsection{Efficient Marginal Runtime}
\label{sec:ordering:oracle}

Instantiating this protocol means running a ranking algorithm repeatedly
(as replicas report additional information).
Rather than recompute the streamed ordering graph and the status of every edge in each invocation,
Algorithm \ref{alg:stream_rpv} could use the output of a past invocation, restricted
to the decisions that were \textit{determinate}, as an oracle.
Corollary \ref{cor:stream_mono} implies that 
\textit{determinate} choices are consistent with any extension of the input,
so this oracle does not change the output.
However, the runtime now depends only on the number of new edges in the streamed ordering graph (\S \ref{sec:oracle_impl}).

\subsection{Protocol Liveness}
\label{sec:liveness}

Algorithm \ref{alg:stream_rpv} is not asymptotically live.  Appendix \ref{apx:non_live}
gives a worst-case example.
The problem is the tiebreaking rule for how Ranked Pairs iterates over
edges of equal weight.
That construction uses a worst-case tiebreaking rule, 
but Ranked Pairs requires only some tiebreaking rule.
Our algorithm can generate a tiebreaking rule dynamically. 
Specifically, when the algorithm visits edges of weight $\gamma=\frac{k}{n}$, 
it defers edges that would become indeterminate
to the end of the ordering among edges of that weight.

\iffalse
As written, an algorithm could vacuously
satisfy the correctness condition of Theorem \ref{thm:output_correct} by refusing to output anything (and i.e. leaving
every edge indeterminate).  Ideally, there would be some bound on the delay between when a client sends a transaction
and when the protocol includes the transaction in its output ordering.

To derive such a bound, we will take advantage of some flexibility in the ordering in which Streaming Ranked Pairs visits edges.
\S \ref{sec:streaming_alg} assumed that there was some predefined total ordering on all transactions, which the streaming algorithm
must respect.  However, any tiebreaking rule when sorting edges by weight suffices to produce correct output of Ranked Pairs.
Our streaming algorithm, therefore, can generate a tiebreaking rule dynamically, in the case of ties by weight.
\fi

Algorithm \ref{alg:stream_rpv_2} thus constructs a (partial) ordering between edges that breaks ties between edges of the same weight.  
The output of this algorithm is, by the same arguments as in \S \ref{sec:streaming_alg}, consistent with the output of non-streaming Ranked Pairs, 
when using an tiebreaking rule consistent with this edge ordering.
The oracle mechanism of
\S \ref{sec:ordering:oracle} passes this tiebreaking information from one run
of Algorithm \ref{alg:stream_rpv_2}
to the next.

\begin{lemmarep}
Whenever Algorithm \ref{alg:stream_rpv_2} includes a \textit{determinate} edge in $H$, or excludes an edge from $H$,
the non-streaming Ranked Pairs, when using the implied ordering, would make the same decision.
%that edge is included (resp. excluded) in the output of the non-streaming Ranked Pairs, when using the implied ordering.
\end{lemmarep}

\begin{proof}
The argument proceeds as in Lemma \ref{lemma:same_choices}.  
Algorithm \ref{alg:stream_rpv_2} induces a tiebreaking rule between edges of the same weight.
If Algorithm \ref{alg:stream_rpv} were given this tiebreaking rule,
it would visit edges in the order in which Algorithm \ref{alg:stream_rpv_2} makes decisions on edges
(i.e., skipping
the instances where Algorithm \ref{alg:stream_rpv_2} defers an edge until later).
%implied by the ordering induced by Algorithm \ref{alg:stream_rpv_2}, it would visit edges in the same ordering as Algorithm \ref{alg:stream_rpv_2}.  
As such, its output is consistent with the non-streaming Ranked Pairs using this implied ordering (by Theorem \ref{thm:output_correct}).
\safeqed
\end{proof}

This small change enables the following structural lemma.  We start with a useful definition.

\begin{definition}
\label{defn:contemporary}
An edge $(\hat{\tx}, \hat{\txp})$ is \textit{contemporary} to an edge $(\tx, \txp)$ if and only if
$\hat{\tx}$ and $\hat{\txp}$ are contained in $U_{\tx, \txp}$.
\end{definition}

%The contemporary relationship is not symmetric.

\begin{lemmarep}
\label{lemma:indet_pair}
If Algorithm \ref{alg:stream_rpv_2} marks an edge $(\tx_i, \tx_j)$ with weight $\frac{k}{n}$ as \textit{indeterminate},
there must be a path from $\tx_j$ to $\tx_i$ contained in $U_{\tx_i, \tx_j}$ that contains an \textit{indeterminate} edge of weight at least 
$\frac{k+1}{n}$.
\end{lemmarep}

\begin{proof}
By construction of the algorithm, if an edge $(\tx_i, \tx_j)$ is marked as \textit{indeterminate},
then there must be a path in $U_{\tx_i, \tx_j}$ already chosen (when this edge is visited)
of \textit{determinate} and \textit{indeterminate} edges of weight at least $\frac{k}{n}$.  Furthermore,
any \textit{indeterminate} edge must be of weight strictly more than $\frac{k}{n}$.  When the algorithm looks for such a path,
it has not yet marked any edges of weight $\frac{k}{n}$ as \textit{indeterminate}, and if no such path exists, then the edge
would not be marked as \textit{indeterminate}.
\safeqed
\end{proof}

By contrast, the best bound in Algorithm \ref{alg:stream_rpv} is that there exists a contemporary indeterminate edge of weight 
$\frac{k}{n}$.
Lemma \ref{lemma:indet_pair} implies that for any indeterminate edge, there is sequence of contemporary
indeterminate edges of strictly increasing weights that causes the edge to be indeterminate.

\begin{lemmarep}
\label{lemma:indet_chain}
Suppose Algorithm \ref{alg:stream_rpv_2} marks an edge $e_1$ of weight $\gamma$ as \textit{indeterminate}.
Then there exists a sequence of \textit{indeterminate} edges $e_1,\ldots e_k$
of strictly increasing weight
such that for each pair of edges $e_i=(\tx_i, \tx_j)$ and $e_{i+1} = (\txp_i, \txp_j)$ with $e_{i+1}$ contemporary to $e_i$,
$\txp_i$ and $\txp_j$ are present in $U_{\tx_i, \tx_j}$,
and the source of $e_k$ is $\hat{v}$.
\end{lemmarep}

\begin{proof}
Follows from repeated application of Lemma \ref{lemma:indet_pair}. 

The algorithm starts with only edges leaving the ``future'' vertex $\hat{v}$ marked as \textit{indeterminate}.
These initial edges are the only \textit{indeterminate} edges with weight $1$, and thus are the only ones
that can form the root of any chain.
\safeqed
\end{proof}

None of these chains can be longer than $n$ edges,
as each step increases the weight by at least $\frac{1}{n}$.
This bound implies that Algorithm \ref{alg:stream_rpv_2} eventually outputs every transaction.

%Note that the maximum length of such a chain is $n$, 
%as weights are bounded between $0$ and $1$, and each step increases the weight by at least $\frac{1}{n}$.

%A similar argument shows the asymptotic liveness of Algorithm \ref{alg:stream_rpv_2}.  Note that asymptotic liveness here
%is a fundamental property of Algorithm \ref{alg:stream_rpv_2}, and does not rely on Assumption \ref{ass:sync}.

\begin{theoremrep}
\label{thm:asymptotic_liveness}
Algorithm \ref{alg:stream_rpv_2} is asymptotically live.
\end{theoremrep}

\begin{proof}

Let $(\hat{\sigma_1},\dots,\hat{\sigma_n})$ be any set of (countably infinite) ordering preferences,
and let $\tx$ be any transaction in those orderings.

Consider the set of all chains of edges $(e_1,\ldots, e_k)$ where $e_{i+1}$ is contemporary to $e_i$ (and of higher weight)
and $e_1$ is adjacent to $\tx$, and $k \leq n$.
 Because each $U_{\tx, \txp}$ is finite,
the number of such chains must be finite,
and the number of transactions that appear in any of these chains must be finite.

On any input to Algorithm \ref{alg:stream_rpv_2} that is a finite truncation of $(\hat{\sigma_1},\dots,\hat{\sigma_n})$,
if an edge adjacent to $\tx$ is left \textit{indeterminate},
then Lemma \ref{lemma:indet_chain} implies that there must be a chain of \textit{indeterminate} edges on this input
with the first adjacent to $\tx$, the last adjacent to $\hat{v}$, and each edge contemporary to the previous.
This chain, with the last edge dropped, must be one of the chains considered above.
As such, there are only a finite number of transactions that could be in the last edge adjacent to $\hat{v}$.
For each such transaction $\txp$, there can only be an edge from $\hat{v}$ to $\txp$ in the streamed ordering graph
if there is some other transaction in the input that is ahead of $\txp$ in at least one replica's vote.
There can only be a finite number of these transactions.

Taking a union over a finite number of finite sets gives a finite set of transactions.  There must therefore be some bound $M_{\tx}$ such
that all of these transactions appear in every replica's vote if $(\hat{\sigma_1},\dots,\hat{\sigma_n})$ is 
truncated to (at least) the first $M_{\tx}$ elements.

A transaction $\tx$ is not output by Algorithm \ref{alg:stream_rpv_2} if the topological sort in the last step
puts a transaction $\txp$ adjacent to an \textit{indeterminate} edge ahead of $\tx$.
But there are only a finite number of transactions that might ever be ahead of $\tx$ (specifically, $P_{\tx}$ and $Q_{\tx}$).

Let $N_\tx=\max_{\txp\in P_{\tx}\cup Q_{\tx}} M_{\txp}$.  Then if $(\hat{\sigma_1},\dots,\hat{\sigma_n})$ is truncated to 
(at least) the first $N_{\tx}$,
then $\tx$ must appear in the output of Algorithm \ref{alg:stream_rpv_2}.
\end{proof}

Additionally, given a network assumption like Assumption \ref{ass:sync}, 
Lemma \ref{lemma:indet_chain} bounds the temporal delay between when a transaction 
is sent by a client and when Algorithm \ref{alg:stream_rpv_2} adds it to its output.
% Of course,
%the delay bound in Assumption \ref{ass:sync} depends on the protocol by which replicas communicate.

\begin{assumption}[Network Delay Bound]
\label{ass:sync}
If a client sends a transaction $\tx$ at time $t$, all replicas include $\tx$ in
their ordering votes before time $t+\Delta$.
%and algorithm receives the vote on the ordering of $\tx$ from each replica before time $t+\Delta$.
\end{assumption}

%No replica can vote on a transaction before a client sends it (we allow a replica to send transactions as a client).
As such, replicas disagree on the ordering between $\tx$ and $\txp$ 
only if they were sent at roughly the same time.
Note that Algorithm \ref{alg:stream_rpv_2} does not require knowledge of $\Delta$, and does not depend on Assumption \ref{ass:sync} for correctness.

\begin{observationrep}
\label{obs:adjacenttimes}
Consider any two transactions $\tx_i$ and $\tx_j$
sent at times $t_i$ and $t_j$, respectively.

If $t_i + \Delta < t_j$, then the weight of the edge $(\tx_i, \tx_j)$ is $1$.
\end{observationrep}
\begin{proof}
If $t_i + \Delta < t_j$, then every replica receives and commits to a vote on $\tx_i$ before $\tx_j$ is sent,
so $\tx_j$ must come after $\tx_i$ in every replica's vote. \safeqed
\end{proof}

%As such, if replicas disagree on the ordering between transactions,
%then they must have been sent at approximately the same time.

\begin{lemmarep}
\label{lemma:onedelta}
Consider any three transactions $\tx$, $\tx_i$, and $\tx_j$ received by all replicas,
that were sent by clients at times $t$, $t_i$, and $t_j$, respectively.
Suppose $\tx\in U_{\tx_i, \tx_j}$, and that the weight of $(\tx_i, \tx_j)$ is not $0$.

Then $t_j - \Delta \leq t \leq t_j + 2\Delta$.
\end{lemmarep}

\begin{proof}
By construction, $\tx\notin R_{\tx_i}$, so $t \leq t_i + \Delta$.  Furthermore, $\tx \notin P_{\tx_j}$, so
$t +\Delta \geq t_j$.
Furthermore, by Observation \ref{obs:adjacenttimes}, $t_j + \Delta \geq t_i$.

Thus, $t_j-\Delta \leq t \leq t_i + \Delta \leq t_j + 2\Delta$.
\safeqed
\end{proof}

This lemma and Lemma \ref{lemma:indet_chain} give an overall time bound.

\begin{theoremrep}
\label{thm:liveness}
A transaction $\tx$ is contained in the output of the algorithm of Algorithm \ref{alg:stream_rpv_2}
after at most $(n+1)\Delta$ time.
\end{theoremrep}

\begin{proof}
Let the current time be $T$, and let $\tx$ be sent at time $t$.  Without loss of generality, assume $\tx$ has been voted on by
every replica.

If an \textit{indeterminate} edge exists to $\tx$, then there exists some transaction $\txp$ (sent at $t^\prime$)
that has not been voted on by every replica, but which preceeds $\tx$ in some replicas' votes.
As such, $T - \Delta \leq t^\prime$.  Then, as $\tx$ does not fully preceed $\txp$, it must be the case that $t^\prime - \Delta \leq t$
(so $T-2\Delta \leq t)$.

Lemma \ref{lemma:indet_chain} implies a sequence of \textit{indeterminate} edges of strictly increasing weight from this edge
to one adjacent to $\hat{v}$.

Consider two adjacent edges $(\tx_i, \tx_j), (\txp_i, \txp_j)$ in this chain, send at times $t_i, t_j, t_i^\prime, t_j^\prime$, respectively.
By Lemma \ref{lemma:onedelta}, it must be the case that $t_j-\Delta \leq t_j^\prime$.
Adding this bound over each link in the chain
shows that
$T-2\Delta - (k-1)\Delta - \Delta \leq t$ (where the last $\Delta$ comes from bounding the 
delay of the transaction in the last edge adjacent to $\hat{v}$).

A sequence can be of length at most $n-1$.

As such, after time $(n+1)\Delta$, every edge adjacent to $\tx$ is either rejected or included and marked \textit{determinate}.

To ensure that $\tx$ is included in the output of Algorithm \ref{alg:stream_rpv_2}, it suffices to ensure that
every transaction adjacent to an \textit{indeterminate} edge
must come after $\tx$.
This is guaranteed by waiting for an additional $\Delta$ time. \safeqed
\end{proof}

\subsection{Trading Accuracy for Liveness}

The $O(n\Delta)$ bound in Theorem \ref{thm:liveness} comes from the fact that a chain of indeterminate edges
can have $O(n)$ length, which, in turn, follows from the fact that weights have a granularity of
$\frac{1}{n}$.  Reducing this granularity by rounding weights reduces the maximum length of a chain of indeterminate 
edges.%thereby reducing the bound on the delay to output a transaction at the cost of reducing the accuracy of the output ordering.

\begin{lemmarep}
If edge weights are rounded to the nearest $\frac{1}{k}$ in the streamed ordering graph
(before Algorithm \ref{alg:stream_rpv_2} is applied)
then
a transaction $\tx$ is contained in the output
after at most $(k+2)\Delta$ time.
\end{lemmarep}

\begin{proof}
The argument proceeds exactly as in that of Theorem \ref{thm:liveness},
except that the length of the chain of \textit{indeterminate} edges is at most $k$.
\safeqed
\end{proof}

However, rounding weakens fairness guarantees.

\begin{lemmarep}
Rounding edge weights to the nearest $\frac{1}{k}$ before applying Algorithm \ref{alg:stream_rpv_2}
achieves $(\gamma, \frac{f}{n} + \frac{1}{2k})$-minimal-batch-order-fairness for all $\gamma$ and $f$ faulty replicas.
\end{lemmarep}

\begin{proof}
The argument proceeds as in Theorem \ref{thm:rp_fair}.  However, due to the rounding and the faulty replicas,
the weight on an edge observed by the algorithm might be up to $\frac{f}{n} + \frac{1}{2k}$ different from the true observed weight
(as opposed to $\frac{f}{n}$, as in Theorem \ref{thm:rp_fair}).
\safeqed
\end{proof}

\section{Related Work}

\subsection{Order-Fairness}

Closest to this work are 
Aequitas \cite{kelkar2020order} and Themis \cite{themis}, which propose the definition of $\gamma$-batch-order-fairness
and instantiate ordering protocols based on aggregating ordering votes into a total order.
Aequitas achieves $(\gamma, \frac{f}{n})$-minimal-batch-order-fairness
for a fixed choice of $\gamma$, but is not asymptotically live.
\citet{aequitaspermissionless} instantiate Aequitas in a setting where the number of replicas is not known,
and similarly points out the connection between this streaming problem and the classical ranking aggregation problem.
 \S \ref{sec:priorcomparison} gives a more detailed comparison to these works.

\citet{cachin2022quick} give a notion of 
``$\kappa$-differential-order-fairness;'' Theorem D.1 of \cite{themis} shows equivalence 
of this definition and $\gamma$-batch-order-fairness, for a choice of $\gamma$ (as a function of $\kappa$).

Our definition of $(\gamma, \delta)$-minimal-batch-order-fairness (Definition \ref{defn:minimalbatch}) strengthens
the (vacuously satisfiable) $\gamma$-batch-order-fairness definition used in prior work (see \S \ref{sec:orderdefn} for details).
%\S \ref{sec:orderdefn} discusses the details of this relationship.

Concurrently with this work, \citet{vafadar2023condorcet}
observe that an adversary can influence the output ordering of many protocols
(including that of \citet{kelkar2020order})
by sending extra transactions (which can force batches to be combined together).
This eliminates the constraints on the output that Definition \ref{defn:gammabatch} gives
(using an example akin to Example \ref{ex:lowgamma}).
Lemma \ref{lemma:inverted} identifies a similar problem.
This observation is a corollary of the fact that the protocol of \citet{kelkar2020order}
(as well as the Ranked Pairs algorithm studied here)
does not satisfy the Irrelevance of Independent Alternatives axiom \cite{arrow1950difficulty}.

%Arrow's impossibility theorem shows no ordering method can be immune to manipulation
%via sending extra transactions (if the method is non-dictatorial and follows unanimous decisions) \cite{arrow1950difficulty}.

\citet{vafadar2023condorcet} also propose using Ranked Pairs to sort transactions within batches 
(as, e.g., those output by a protocol like Themis \cite{themis}).
However,
no modification to the ordering of transactions within a batch output in Aequitas \cite{kelkar2020order}
can give the Ranked Pairs ordering, because Ranked Pairs might need to interleave batches that Aequitas considers incomparable
(as demonstrated in \S \ref{apx:interleave}).  This strategy also cannot mitigate Aequitas's lack of asymptotic liveness.
This mitigation, if applied to Aequitas, is akin to dropping all edges of weight less than $\gamma$ before running Ranked Pairs.
This strategy, therefore, would achieve $(\gamma^\prime, \frac{f}{n})$-minimal-batch-order-fairness for all $\gamma^\prime \geq \gamma$ (a weaker
guarantee than what our protocol provides).

Also concurrently with this work, 
\citet{kiayias2024ordering} propose a different method for ordering transactions within batches 
(but rely on a specific choice of a parameter $\gamma$ to determine those batches).  A version
of that protocol adds timestamps on reported orderings to improve liveness.

As discussed in \S\ref{sec:priorcomparison}, the batches in Themis \cite{themis} are defined slightly differently than in Definition \ref{defn:gammabatch}.
Themis, when computing a graph of ordering dependencies, adds a directed edge between every pair of vertices, choosing the direction of an edge by majority vote (ignoring the parameter $\gamma$, which is considered only on edges adjacent to transactions that some replicas have not yet reported observing).
In so doing, it merges together any two batches unless, for any transaction $\tx$ in one and $\txp$ in the other,
a majority of replicas vote $\tx \before \txp$ (which bypasses the counterexample of \S \ref{apx:interleave}).

Additionally, Themis's \cite{themis} round-based sequencing protocol presents a different problem insurmountable
by any process like the one suggested by \citet{vafadar2023condorcet}
that sorts batches separately.
%However, Themis's round-based sequencing protocol presents a different problem insurmountable by any protocol 
%as suggested in \cite{vafadar2023condorcet} that sorts batches separately.
Namely, Themis may finalize the composition of a batch before  has enough information to guarantee
that, according to the Ranked Pairs ordering, all transactions in the batch must come before
all those not in the batch (and not already in the output).
Thus, it may output two batches in sequence, such that running Ranked Pairs on each batch separately
produces a different output than running Ranked Pairs on the union of the two batches together.
%A connected component (considering edges of weight at least $1/2$)
%could be split into two (or more) rounds and thus two (or more) batches, meaning that Ranked Pairs run on each batch
%separately may not produce the same result as Ranked Pairs would when run on the whole component.
We give a example of this phenomenon in \S \ref{apx:early_finalize}.

%In the specific construction, a modification of Themis \cite{themis} (Remark 3, \cite{vafadar2023condorcet}), this would not output the Ranked Pairs
%ordering on the overall set of transactions (see Example \ref{ex:badbatch} for details),
%nor would an analogous construction that runs Ranked Pairs to sort within the batches of Aequitas \cite{kelkar2020order}.

\citet{zhang2020byzantine}
add timestamps to ordering votes and sorts transactions by median timestamp to provide a linearizability guarantee,
which is incomparable to $(\gamma, \delta)$-minimal-batch-order-fairness.  \citet{mamageishvili2023buying} propose letting users
pay to reduce their reported timestamps.

\subsection{Social Choice}
Classical work of social choice studies the problem of choosing an order between a fixed, finite set of
candidates given a complete set of preferences from each voter \cite{arrow1950difficulty}.
Prior work studies related models in a limited-information setting.
Lu and Boutilier \cite{lu2013multi}, Ackerman et. al \cite{ackerman2013elections},
and Cullinan et. al \cite{cullinan2014borda} study the setting where a social choice rule has incomplete access (a partial ordering)
to a voter's preferences over a finite set of candidates.
Conitzer and Sandholm \cite{conitzer2005communication}
study the communication complexity of a variety of voting rules.
Fain et. al \cite{fain2019random} study a randomized ordering rule
that requires only a constant number of queries to voter preferences.

Fishburn \cite{fishburn1970arrow} showed that Arrow's impossibility theorem does not hold
when the set of voters is infinite, although Kirman and Sondermann \cite{kirman1972arrow}
find ``dictatorial sets'' of voters.  Grafe and Grafe \cite{grafe1983arrow}
extend these results to the case of infinitely many alternatives,
subject to a continuity condition
on the space of ranking preferences.
Chichilnisky and Heal \cite{chichilnisky1997social}
and Efimov and Koshevoy \cite{efimov1994topological} study the types of rules admissible in the infinite voters setting,
given a topology on ordering preferences.

%\citet{crisman2022voting} study the preference aggregation problem where the output
% must be a cyclic permutation of the candidates.

\subsection{Front-Running}

The study of order-fairness is often motivated by the goal of limiting
 an adversary's ability to order transactions (``front-running'') in public blockchains (as in,
 e.g., \citet{daian2020flash}).
Additional approaches to this problem include commiting to an ordering
before any replica can know transaction contents, using threshold encryption or commit-reveal schemes
\cite{malkhi2022maximal,clineclockwork,zhang2022flash}. 
Li et. al \cite{li2023transaction} study the problem of verifiably computing an 
ordering within a trusted hardware enclave.
These could be applied on top of a streaming ordering algorithm.

Kavousi et. al \cite{kavousi2023blindperm} randomly shuffle blocks of encrypted transactions.
Some decentralized exchanges \cite{ramseyer2023speedex,cowswapproblem,penumbraswap} process transactions in unordered batches,
eliminating the need for an ordering within a block (but not eliminating all order manipulations \cite{zhangcomputation}).
\citet{ferreira2022credible} and \citet{li2023mev} construct sequencing rules specific to decentralized exchanges.
%\citet{mamageishvili2023buying} use a median-timestamp mechanism akin to \citet{zhang2020byzantine}, but allow clients to pay to  reduce its reported timestamp.

% Fain et. al \cite{fain2019random} 

%\XXX{citation to condorcet, related social choice papers?}

\section{Conclusion}

We study here the problem of ordering transactions in a replicated state machine
as a novel streaming instance of the preference aggregation problem from classical social choice theory.
This viewpoint enables us to strengthen the notions of ``order-fairness'' used in prior work,
and then construct algorithms solving this problem with both strictly stronger ``fairness'' guarantees
and strictly stronger liveness properties than all of the prior work.
To be specific, our streaming variant of Ranked Pairs satisfies 
$(\gamma, \frac{f}{n})$-minimal-batch-order-fairness
for every $\frac{1}{2}<\gamma\leq 1$ simultaneously, and for any number of faulty replicas $f$.
Fairness guarantees smoothly weaken as the number of faulty replicas increases.
For comparison, prior work must fix a choice of $\gamma$ 
 and a bound on the number of faulty replicas in advance, and can only satisfy $\gamma$-batch-order-fairness for that $\gamma$.

These notions of ``order-fairness'' are not the only desiderata with practical significance.
Different contexts in which a system is deployed pose different constraints and financial incentives.
We believe that social-choice style methods for preference aggregation raises many interesting open questions
of practical import.
For example, what aggregation rules lead to high social welfare?  If clients bribe replicas to prefer certain orderings,
what rules maximize (or minimize) this fee revenue?
How do incentives distort the output ordering?

\if\acm1
\bibliographystyle{ACM-Reference-Format}
\fi

\if\llncs1
\bibliographystyle{splncs04}
\fi

\if\ieee1
\bibliographystyle{IEEEtran}
\fi

\bibliography{paper,paper_unique}

%%% -*-BibTeX-*-
%%% Do NOT edit. File created by BibTeX with style
%%% ACM-Reference-Format-Journals [18-Jan-2012].

\begin{thebibliography}{45}

%%% ====================================================================
%%% NOTE TO THE USER: you can override these defaults by providing
%%% customized versions of any of these macros before the \bibliography
%%% command.  Each of them MUST provide its own final punctuation,
%%% except for \shownote{}, \showDOI{}, and \showURL{}.  The latter two
%%% do not use final punctuation, in order to avoid confusing it with
%%% the Web address.
%%%
%%% To suppress output of a particular field, define its macro to expand
%%% to an empty string, or better, \unskip, like this:
%%%
%%% \newcommand{\showDOI}[1]{\unskip}   % LaTeX syntax
%%%
%%% \def \showDOI #1{\unskip}           % plain TeX syntax
%%%
%%% ====================================================================

\ifx \showCODEN    \undefined \def \showCODEN     #1{\unskip}     \fi
\ifx \showDOI      \undefined \def \showDOI       #1{#1}\fi
\ifx \showISBNx    \undefined \def \showISBNx     #1{\unskip}     \fi
\ifx \showISBNxiii \undefined \def \showISBNxiii  #1{\unskip}     \fi
\ifx \showISSN     \undefined \def \showISSN      #1{\unskip}     \fi
\ifx \showLCCN     \undefined \def \showLCCN      #1{\unskip}     \fi
\ifx \shownote     \undefined \def \shownote      #1{#1}          \fi
\ifx \showarticletitle \undefined \def \showarticletitle #1{#1}   \fi
\ifx \showURL      \undefined \def \showURL       {\relax}        \fi
% The following commands are used for tagged output and should be
% invisible to TeX
\providecommand\bibfield[2]{#2}
\providecommand\bibinfo[2]{#2}
\providecommand\natexlab[1]{#1}
\providecommand\showeprint[2][]{arXiv:#2}

\bibitem[cow({[n.\,d.]})]%
        {cowswapproblem}
 \bibinfo{year}{[n.\,d.]}\natexlab{}.
\newblock \bibinfo{title}{CoW Protocol Overview: The Batch Auction Optimization
  Problem}.
\newblock
  \bibinfo{howpublished}{\url{https://web.archive.org/web/20220614183101/https://docs.cow.fi/off-chain-services/in-depth-solver-specification/the-batch-auction-optimization-problem}}.
\newblock
\newblock
\shownote{Accessed 10/19/2022}.


\bibitem[pen({[n.\,d.]})]%
        {penumbraswap}
 \bibinfo{year}{[n.\,d.]}\natexlab{}.
\newblock \bibinfo{title}{The Penumbra Protocol: Sealed-Bid Batch Swaps}.
\newblock
  \bibinfo{howpublished}{\url{https://web.archive.org/web/20220614034906/https://protocol.penumbra.zone/main/zswap/swap.html}}.
\newblock
\newblock
\shownote{Accessed 10/19/2022}.


\bibitem[Ackerman et~al\mbox{.}(2013)]%
        {ackerman2013elections}
\bibfield{author}{\bibinfo{person}{Michael Ackerman},
  \bibinfo{person}{Sul-Young Choi}, \bibinfo{person}{Peter Coughlin},
  \bibinfo{person}{Eric Gottlieb}, {and} \bibinfo{person}{Japheth Wood}.}
  \bibinfo{year}{2013}\natexlab{}.
\newblock \showarticletitle{Elections with partially ordered preferences}.
\newblock \bibinfo{journal}{\emph{Public Choice}}  \bibinfo{volume}{157}
  (\bibinfo{year}{2013}), \bibinfo{pages}{145--168}.
\newblock


\bibitem[Arrow(1950)]%
        {arrow1950difficulty}
\bibfield{author}{\bibinfo{person}{Kenneth~J Arrow}.}
  \bibinfo{year}{1950}\natexlab{}.
\newblock \showarticletitle{A difficulty in the concept of social welfare}.
\newblock \bibinfo{journal}{\emph{Journal of political economy}}
  \bibinfo{volume}{58}, \bibinfo{number}{4} (\bibinfo{year}{1950}),
  \bibinfo{pages}{328--346}.
\newblock


\bibitem[Cachin et~al\mbox{.}(2022)]%
        {cachin2022quick}
\bibfield{author}{\bibinfo{person}{Christian Cachin}, \bibinfo{person}{Jovana
  Mi{\'c}i{\'c}}, \bibinfo{person}{Nathalie Steinhauer}, {and}
  \bibinfo{person}{Luca Zanolini}.} \bibinfo{year}{2022}\natexlab{}.
\newblock \showarticletitle{Quick order fairness}. In
  \bibinfo{booktitle}{\emph{Financial Cryptography and Data Security: 26th
  International Conference, FC 2022, Grenada, May 2--6, 2022, Revised Selected
  Papers}}. Springer, \bibinfo{pages}{316--333}.
\newblock


\bibitem[Castro et~al\mbox{.}(1999)]%
        {castro1999practical}
\bibfield{author}{\bibinfo{person}{Miguel Castro}, \bibinfo{person}{Barbara
  Liskov}, {et~al\mbox{.}}} \bibinfo{year}{1999}\natexlab{}.
\newblock \showarticletitle{Practical byzantine fault tolerance}. In
  \bibinfo{booktitle}{\emph{OsDI}}, Vol.~\bibinfo{volume}{99}.
  \bibinfo{pages}{173--186}.
\newblock


\bibitem[Chichilnisky and Heal(1997)]%
        {chichilnisky1997social}
\bibfield{author}{\bibinfo{person}{Graciela Chichilnisky} {and}
  \bibinfo{person}{Geoffrey Heal}.} \bibinfo{year}{1997}\natexlab{}.
\newblock \showarticletitle{Social choice with infinite populations:
  construction of a rule and impossibility results}.
\newblock \bibinfo{journal}{\emph{Social Choice and Welfare}}
  \bibinfo{volume}{14} (\bibinfo{year}{1997}), \bibinfo{pages}{303--318}.
\newblock


\bibitem[Cline et~al\mbox{.}({[n.\,d.]})]%
        {clineclockwork}
\bibfield{author}{\bibinfo{person}{Dan Cline}, \bibinfo{person}{Thaddeus
  Dryja}, {and} \bibinfo{person}{Neha Narula}.}
  \bibinfo{year}{[n.\,d.]}\natexlab{}.
\newblock \showarticletitle{ClockWork: An Exchange Protocol for Proofs of Non
  Front-Running}.
\newblock  (\bibinfo{year}{[n.\,d.]}).
\newblock


\bibitem[Colomer(2013)]%
        {colomer2013ramon}
\bibfield{author}{\bibinfo{person}{Josep~M Colomer}.}
  \bibinfo{year}{2013}\natexlab{}.
\newblock \showarticletitle{Ramon Llull: from ‘Ars electionis’ to social
  choice theory}.
\newblock \bibinfo{journal}{\emph{Social Choice and Welfare}}
  \bibinfo{volume}{40}, \bibinfo{number}{2} (\bibinfo{year}{2013}),
  \bibinfo{pages}{317--328}.
\newblock


\bibitem[Condorcet and Caritat(1785)]%
        {condorcet1785essay}
\bibfield{author}{\bibinfo{person}{Marquis~de Condorcet} {and}
  \bibinfo{person}{Marquis~de Caritat}.} \bibinfo{year}{1785}\natexlab{}.
\newblock \showarticletitle{An Essay on the Application of Analysis to the
  Probability of Decisions Rendered by a Plurality of Votes}.
\newblock \bibinfo{journal}{\emph{Classics of social choice}}
  (\bibinfo{year}{1785}), \bibinfo{pages}{91--112}.
\newblock


\bibitem[Conitzer and Sandholm(2005)]%
        {conitzer2005communication}
\bibfield{author}{\bibinfo{person}{Vincent Conitzer} {and}
  \bibinfo{person}{Tuomas Sandholm}.} \bibinfo{year}{2005}\natexlab{}.
\newblock \showarticletitle{Communication complexity of common voting rules}.
  In \bibinfo{booktitle}{\emph{Proceedings of the 6th ACM conference on
  Electronic commerce}}. \bibinfo{pages}{78--87}.
\newblock


\bibitem[Cullinan et~al\mbox{.}(2014)]%
        {cullinan2014borda}
\bibfield{author}{\bibinfo{person}{John Cullinan}, \bibinfo{person}{Samuel~K
  Hsiao}, {and} \bibinfo{person}{David Polett}.}
  \bibinfo{year}{2014}\natexlab{}.
\newblock \showarticletitle{A Borda count for partially ordered ballots}.
\newblock \bibinfo{journal}{\emph{Social Choice and Welfare}}
  \bibinfo{volume}{42} (\bibinfo{year}{2014}), \bibinfo{pages}{913--926}.
\newblock


\bibitem[Daian et~al\mbox{.}(2020)]%
        {daian2020flash}
\bibfield{author}{\bibinfo{person}{Philip Daian}, \bibinfo{person}{Steven
  Goldfeder}, \bibinfo{person}{Tyler Kell}, \bibinfo{person}{Yunqi Li},
  \bibinfo{person}{Xueyuan Zhao}, \bibinfo{person}{Iddo Bentov},
  \bibinfo{person}{Lorenz Breidenbach}, {and} \bibinfo{person}{Ari Juels}.}
  \bibinfo{year}{2020}\natexlab{}.
\newblock \showarticletitle{Flash boys 2.0: Frontrunning in decentralized
  exchanges, miner extractable value, and consensus instability}. In
  \bibinfo{booktitle}{\emph{2020 IEEE Symposium on Security and Privacy (SP)}}.
  IEEE, \bibinfo{pages}{910--927}.
\newblock


\bibitem[D\'{e}fago et~al\mbox{.}(2004)]%
        {broadcastsurvey}
\bibfield{author}{\bibinfo{person}{Xavier D\'{e}fago},
  \bibinfo{person}{Andr\'{e} Schiper}, {and} \bibinfo{person}{P\'{e}ter
  Urb\'{a}n}.} \bibinfo{year}{2004}\natexlab{}.
\newblock \showarticletitle{Total order broadcast and multicast algorithms:
  Taxonomy and survey}.
\newblock \bibinfo{journal}{\emph{ACM Comput. Surv.}} \bibinfo{volume}{36},
  \bibinfo{number}{4} (\bibinfo{date}{dec} \bibinfo{year}{2004}),
  \bibinfo{pages}{372–421}.
\newblock
\showISSN{0360-0300}
\urldef\tempurl%
\url{https://doi.org/10.1145/1041680.1041682}
\showDOI{\tempurl}


\bibitem[Efimov and Koshevoy(1994)]%
        {efimov1994topological}
\bibfield{author}{\bibinfo{person}{Boris~A Efimov} {and}
  \bibinfo{person}{Gleb~A Koshevoy}.} \bibinfo{year}{1994}\natexlab{}.
\newblock \showarticletitle{A topological approach to social choice with
  infinite populations}.
\newblock \bibinfo{journal}{\emph{Mathematical Social Sciences}}
  \bibinfo{volume}{27}, \bibinfo{number}{2} (\bibinfo{year}{1994}),
  \bibinfo{pages}{145--157}.
\newblock


\bibitem[Fain et~al\mbox{.}(2019)]%
        {fain2019random}
\bibfield{author}{\bibinfo{person}{Brandon Fain}, \bibinfo{person}{Ashish
  Goel}, \bibinfo{person}{Kamesh Munagala}, {and} \bibinfo{person}{Nina
  Prabhu}.} \bibinfo{year}{2019}\natexlab{}.
\newblock \showarticletitle{Random dictators with a random referee: Constant
  sample complexity mechanisms for social choice}. In
  \bibinfo{booktitle}{\emph{Proceedings of the AAAI Conference on Artificial
  Intelligence}}, Vol.~\bibinfo{volume}{33}. \bibinfo{pages}{1893--1900}.
\newblock


\bibitem[Ferreira and Parkes(2022)]%
        {ferreira2022credible}
\bibfield{author}{\bibinfo{person}{Matheus~VX Ferreira} {and}
  \bibinfo{person}{David~C Parkes}.} \bibinfo{year}{2022}\natexlab{}.
\newblock \showarticletitle{Credible Decentralized Exchange Design via
  Verifiable Sequencing Rules}.
\newblock \bibinfo{journal}{\emph{arXiv preprint arXiv:2209.15569}}
  (\bibinfo{year}{2022}).
\newblock


\bibitem[Fishburn(1970)]%
        {fishburn1970arrow}
\bibfield{author}{\bibinfo{person}{Peter~C Fishburn}.}
  \bibinfo{year}{1970}\natexlab{}.
\newblock \showarticletitle{Arrow's impossibility theorem: concise proof and
  infinite voters}.
\newblock \bibinfo{journal}{\emph{Journal of Economic Theory}}
  \bibinfo{volume}{2}, \bibinfo{number}{1} (\bibinfo{year}{1970}),
  \bibinfo{pages}{103--106}.
\newblock


\bibitem[Garay et~al\mbox{.}(2015)]%
        {garay2015bitcoin}
\bibfield{author}{\bibinfo{person}{Juan Garay}, \bibinfo{person}{Aggelos
  Kiayias}, {and} \bibinfo{person}{Nikos Leonardos}.}
  \bibinfo{year}{2015}\natexlab{}.
\newblock \showarticletitle{The bitcoin backbone protocol: Analysis and
  applications}. In \bibinfo{booktitle}{\emph{Advances in Cryptology-EUROCRYPT
  2015: 34th Annual International Conference on the Theory and Applications of
  Cryptographic Techniques, Sofia, Bulgaria, April 26-30, 2015, Proceedings,
  Part II}}. Springer, \bibinfo{pages}{281--310}.
\newblock


\bibitem[Grafe and Grafe(1983)]%
        {grafe1983arrow}
\bibfield{author}{\bibinfo{person}{F Grafe} {and} \bibinfo{person}{J Grafe}.}
  \bibinfo{year}{1983}\natexlab{}.
\newblock \showarticletitle{On Arrow-type impossibility theorems with infinite
  individuals and infinite alternatives}.
\newblock \bibinfo{journal}{\emph{Economics Letters}} \bibinfo{volume}{11},
  \bibinfo{number}{1-2} (\bibinfo{year}{1983}), \bibinfo{pages}{75--79}.
\newblock


\bibitem[H{\"a}gele and Pukelsheim(2001)]%
        {hagele2001lulls}
\bibfield{author}{\bibinfo{person}{G{\"u}nter H{\"a}gele} {and}
  \bibinfo{person}{Friedrich Pukelsheim}.} \bibinfo{year}{2001}\natexlab{}.
\newblock \showarticletitle{Lulls' writings on electoral sytems}.
\newblock  (\bibinfo{year}{2001}).
\newblock


\bibitem[Kavousi et~al\mbox{.}(2023)]%
        {kavousi2023blindperm}
\bibfield{author}{\bibinfo{person}{Alireza Kavousi}, \bibinfo{person}{Duc~V
  Le}, \bibinfo{person}{Philipp Jovanovic}, {and} \bibinfo{person}{George
  Danezis}.} \bibinfo{year}{2023}\natexlab{}.
\newblock \showarticletitle{BlindPerm: Efficient MEV Mitigation with an
  Encrypted Mempool and Permutation}.
\newblock \bibinfo{journal}{\emph{Cryptology ePrint Archive}}
  (\bibinfo{year}{2023}).
\newblock


\bibitem[Kelkar et~al\mbox{.}(2022)]%
        {aequitaspermissionless}
\bibfield{author}{\bibinfo{person}{Mahimna Kelkar}, \bibinfo{person}{Soubhik
  Deb}, {and} \bibinfo{person}{Sreeram Kannan}.}
  \bibinfo{year}{2022}\natexlab{}.
\newblock \showarticletitle{Order-fair consensus in the permissionless
  setting}. In \bibinfo{booktitle}{\emph{Proceedings of the 9th ACM on ASIA
  Public-Key Cryptography Workshop}}. \bibinfo{pages}{3--14}.
\newblock


\bibitem[Kelkar et~al\mbox{.}(2023)]%
        {themis}
\bibfield{author}{\bibinfo{person}{Mahimna Kelkar}, \bibinfo{person}{Soubhik
  Deb}, \bibinfo{person}{Sishan Long}, \bibinfo{person}{Ari Juels}, {and}
  \bibinfo{person}{Sreeram Kannan}.} \bibinfo{year}{2023}\natexlab{}.
\newblock \showarticletitle{Themis: Fast, strong order-fairness in byzantine
  consensus}. In \bibinfo{booktitle}{\emph{Proceedings of the 2023 ACM SIGSAC
  Conference on Computer and Communications Security}}.
  \bibinfo{pages}{475--489}.
\newblock


\bibitem[Kelkar et~al\mbox{.}(2020)]%
        {kelkar2020order}
\bibfield{author}{\bibinfo{person}{Mahimna Kelkar}, \bibinfo{person}{Fan
  Zhang}, \bibinfo{person}{Steven Goldfeder}, {and} \bibinfo{person}{Ari
  Juels}.} \bibinfo{year}{2020}\natexlab{}.
\newblock \showarticletitle{Order-fairness for byzantine consensus}. In
  \bibinfo{booktitle}{\emph{Advances in Cryptology--CRYPTO 2020: 40th Annual
  International Cryptology Conference, CRYPTO 2020, Santa Barbara, CA, USA,
  August 17--21, 2020, Proceedings, Part III 40}}. Springer,
  \bibinfo{pages}{451--480}.
\newblock


\bibitem[Kiayias et~al\mbox{.}(2024)]%
        {kiayias2024ordering}
\bibfield{author}{\bibinfo{person}{Aggelos Kiayias}, \bibinfo{person}{Nikos
  Leonardos}, {and} \bibinfo{person}{Yu Shen}.}
  \bibinfo{year}{2024}\natexlab{}.
\newblock \showarticletitle{Ordering transactions with bounded unfairness:
  definitions, complexity and constructions}. In
  \bibinfo{booktitle}{\emph{Annual International Conference on the Theory and
  Applications of Cryptographic Techniques}}. Springer,
  \bibinfo{pages}{34--63}.
\newblock


\bibitem[Kirman and Sondermann(1972)]%
        {kirman1972arrow}
\bibfield{author}{\bibinfo{person}{Alan~P Kirman} {and} \bibinfo{person}{Dieter
  Sondermann}.} \bibinfo{year}{1972}\natexlab{}.
\newblock \showarticletitle{Arrow's theorem, many agents, and invisible
  dictators}.
\newblock \bibinfo{journal}{\emph{Journal of Economic Theory}}
  \bibinfo{volume}{5}, \bibinfo{number}{2} (\bibinfo{year}{1972}),
  \bibinfo{pages}{267--277}.
\newblock


\bibitem[Lamport(2001)]%
        {lamport2001paxos}
\bibfield{author}{\bibinfo{person}{Leslie Lamport}.}
  \bibinfo{year}{2001}\natexlab{}.
\newblock \showarticletitle{Paxos made simple}.
\newblock \bibinfo{journal}{\emph{ACM SIGACT News (Distributed Computing
  Column) 32, 4 (Whole Number 121, December 2001)}} (\bibinfo{year}{2001}),
  \bibinfo{pages}{51--58}.
\newblock


\bibitem[Li et~al\mbox{.}(2023a)]%
        {li2023transaction}
\bibfield{author}{\bibinfo{person}{Rujia Li}, \bibinfo{person}{Xuanwei Hu},
  \bibinfo{person}{Qin Wang}, \bibinfo{person}{Sisi Duan}, {and}
  \bibinfo{person}{Qi Wang}.} \bibinfo{year}{2023}\natexlab{a}.
\newblock \showarticletitle{Transaction Fairness in Blockchains, Revisited}.
\newblock \bibinfo{journal}{\emph{Cryptology ePrint Archive}}
  (\bibinfo{year}{2023}).
\newblock


\bibitem[Li et~al\mbox{.}(2023b)]%
        {li2023mev}
\bibfield{author}{\bibinfo{person}{Yuhao Li}, \bibinfo{person}{Mengqian Zhang},
  \bibinfo{person}{Jichen Li}, \bibinfo{person}{Elynn Chen},
  \bibinfo{person}{Xi Chen}, {and} \bibinfo{person}{Xiaotie Deng}.}
  \bibinfo{year}{2023}\natexlab{b}.
\newblock \showarticletitle{MEV Makes Everyone Happy under Greedy Sequencing
  Rule}. In \bibinfo{booktitle}{\emph{Proceedings of the 2023 Workshop on
  Decentralized Finance and Security}}. \bibinfo{pages}{9--15}.
\newblock


\bibitem[Llull and Gaya(1978)]%
        {llull}
\bibfield{author}{\bibinfo{person}{Ramon Llull} {and} \bibinfo{person}{Jordi
  Gaya}.} \bibinfo{year}{1978}\natexlab{}.
\newblock \bibinfo{booktitle}{\emph{Ars notatoria}}.
\newblock \bibinfo{publisher}{Citema}.
\newblock


\bibitem[Lokhava et~al\mbox{.}(2019)]%
        {lokhava:stellar}
\bibfield{author}{\bibinfo{person}{Marta Lokhava}, \bibinfo{person}{Giuliano
  Losa}, \bibinfo{person}{David Mazi\`{e}res}, \bibinfo{person}{Graydon Hoare},
  \bibinfo{person}{Nicolas Barry}, \bibinfo{person}{Eli Gafni},
  \bibinfo{person}{Jonathan Jove}, \bibinfo{person}{Rafa\l{} Malinowsky}, {and}
  \bibinfo{person}{Jed McCaleb}.} \bibinfo{year}{2019}\natexlab{}.
\newblock \showarticletitle{Fast and Secure Global Payments with Stellar}. In
  \bibinfo{booktitle}{\emph{Proceedings of the 27th ACM Symposium on Operating
  Systems Principles}} (Huntsville, Ontario, Canada)
  \emph{(\bibinfo{series}{SOSP '19})}. \bibinfo{publisher}{Association for
  Computing Machinery}, \bibinfo{address}{New York, NY, USA},
  \bibinfo{pages}{80–96}.
\newblock
\showISBNx{9781450368735}
\urldef\tempurl%
\url{https://doi.org/10.1145/3341301.3359636}
\showDOI{\tempurl}


\bibitem[Lu and Boutilier(2013)]%
        {lu2013multi}
\bibfield{author}{\bibinfo{person}{Tyler Lu} {and} \bibinfo{person}{Craig
  Boutilier}.} \bibinfo{year}{2013}\natexlab{}.
\newblock \showarticletitle{Multi-winner social choice with incomplete
  preferences}. In \bibinfo{booktitle}{\emph{Twenty-Third International Joint
  Conference on Artificial Intelligence}}.
\newblock


\bibitem[Malkhi and Szalachowski(2022)]%
        {malkhi2022maximal}
\bibfield{author}{\bibinfo{person}{Dahlia Malkhi} {and} \bibinfo{person}{Pawel
  Szalachowski}.} \bibinfo{year}{2022}\natexlab{}.
\newblock \showarticletitle{Maximal Extractable Value (MEV) Protection on a
  DAG}.
\newblock \bibinfo{journal}{\emph{arXiv preprint arXiv:2208.00940}}
  (\bibinfo{year}{2022}).
\newblock


\bibitem[Mamageishvili et~al\mbox{.}(2023)]%
        {mamageishvili2023buying}
\bibfield{author}{\bibinfo{person}{Akaki Mamageishvili},
  \bibinfo{person}{Mahimna Kelkar}, \bibinfo{person}{Jan~Christoph Schlegel},
  {and} \bibinfo{person}{Edward~W Felten}.} \bibinfo{year}{2023}\natexlab{}.
\newblock \showarticletitle{Buying Time: Latency Racing vs. Bidding for
  Transaction Ordering}. In \bibinfo{booktitle}{\emph{5th Conference on
  Advances in Financial Technologies (AFT 2023)}}. Schloss-Dagstuhl-Leibniz
  Zentrum f{\"u}r Informatik.
\newblock


\bibitem[Miller et~al\mbox{.}(2016)]%
        {miller2016honey}
\bibfield{author}{\bibinfo{person}{Andrew Miller}, \bibinfo{person}{Yu Xia},
  \bibinfo{person}{Kyle Croman}, \bibinfo{person}{Elaine Shi}, {and}
  \bibinfo{person}{Dawn Song}.} \bibinfo{year}{2016}\natexlab{}.
\newblock \showarticletitle{The honey badger of BFT protocols}. In
  \bibinfo{booktitle}{\emph{Proceedings of the 2016 ACM SIGSAC conference on
  computer and communications security}}. \bibinfo{pages}{31--42}.
\newblock


\bibitem[Oki and Liskov(1988)]%
        {oki1988viewstamped}
\bibfield{author}{\bibinfo{person}{Brian~M Oki} {and}
  \bibinfo{person}{Barbara~H Liskov}.} \bibinfo{year}{1988}\natexlab{}.
\newblock \showarticletitle{Viewstamped replication: A new primary copy method
  to support highly-available distributed systems}. In
  \bibinfo{booktitle}{\emph{Proceedings of the seventh annual ACM Symposium on
  Principles of distributed computing}}. \bibinfo{pages}{8--17}.
\newblock


\bibitem[Ongaro and Ousterhout(2014)]%
        {ongaro2014search}
\bibfield{author}{\bibinfo{person}{Diego Ongaro} {and} \bibinfo{person}{John
  Ousterhout}.} \bibinfo{year}{2014}\natexlab{}.
\newblock \showarticletitle{In search of an understandable consensus
  algorithm}.
\newblock  (\bibinfo{year}{2014}), \bibinfo{pages}{305--319}.
\newblock


\bibitem[Ramseyer et~al\mbox{.}(2023)]%
        {ramseyer2023speedex}
\bibfield{author}{\bibinfo{person}{Geoffrey Ramseyer}, \bibinfo{person}{Ashish
  Goel}, {and} \bibinfo{person}{David Mazi{\`e}res}.}
  \bibinfo{year}{2023}\natexlab{}.
\newblock \showarticletitle{{SPEEDEX}: A Scalable, Parallelizable, and
  Economically Efficient Decentralized EXchange}. In
  \bibinfo{booktitle}{\emph{20th USENIX Symposium on Networked Systems Design
  and Implementation (NSDI 23)}}. \bibinfo{pages}{849--875}.
\newblock


\bibitem[Tideman(1987)]%
        {tideman1987independence}
\bibfield{author}{\bibinfo{person}{T~Nicolaus Tideman}.}
  \bibinfo{year}{1987}\natexlab{}.
\newblock \showarticletitle{Independence of clones as a criterion for voting
  rules}.
\newblock \bibinfo{journal}{\emph{Social Choice and Welfare}}
  \bibinfo{volume}{4}, \bibinfo{number}{3} (\bibinfo{year}{1987}),
  \bibinfo{pages}{185--206}.
\newblock


\bibitem[Vafadar and Khabbazian(2023)]%
        {vafadar2023condorcet}
\bibfield{author}{\bibinfo{person}{Mohammad~Amin Vafadar} {and}
  \bibinfo{person}{Majid Khabbazian}.} \bibinfo{year}{2023}\natexlab{}.
\newblock \showarticletitle{Condorcet Attack Against Fair Transaction
  Ordering}. In \bibinfo{booktitle}{\emph{5th Conference on Advances in
  Financial Technologies}}.
\newblock


\bibitem[Yin et~al\mbox{.}(2019)]%
        {yin:hotstuff}
\bibfield{author}{\bibinfo{person}{Maofan Yin}, \bibinfo{person}{Dahlia
  Malkhi}, \bibinfo{person}{Michael~K. Reiter}, \bibinfo{person}{Guy~Golan
  Gueta}, {and} \bibinfo{person}{Ittai Abraham}.}
  \bibinfo{year}{2019}\natexlab{}.
\newblock \showarticletitle{HotStuff: BFT Consensus with Linearity and
  Responsiveness}. In \bibinfo{booktitle}{\emph{Proceedings of the 2019 ACM
  Symposium on Principles of Distributed Computing}} (Toronto ON, Canada)
  \emph{(\bibinfo{series}{PODC '19})}. \bibinfo{publisher}{Association for
  Computing Machinery}, \bibinfo{address}{New York, NY, USA},
  \bibinfo{pages}{347–356}.
\newblock
\showISBNx{9781450362177}
\urldef\tempurl%
\url{https://doi.org/10.1145/3293611.3331591}
\showDOI{\tempurl}


\bibitem[Zhang et~al\mbox{.}(2022)]%
        {zhang2022flash}
\bibfield{author}{\bibinfo{person}{Haoqian Zhang}, \bibinfo{person}{Louis-Henri
  Merino}, \bibinfo{person}{Vero Estrada-Galinanes}, {and}
  \bibinfo{person}{Bryan Ford}.} \bibinfo{year}{2022}\natexlab{}.
\newblock \showarticletitle{Flash freezing flash boys: Countering blockchain
  front-running}. In \bibinfo{booktitle}{\emph{2022 IEEE 42nd International
  Conference on Distributed Computing Systems Workshops (ICDCSW)}}. IEEE,
  \bibinfo{pages}{90--95}.
\newblock


\bibitem[Zhang et~al\mbox{.}({[n.\,d.]})]%
        {zhangcomputation}
\bibfield{author}{\bibinfo{person}{Mengqian Zhang}, \bibinfo{person}{Yuhao Li},
  \bibinfo{person}{Xinyuan Sun}, \bibinfo{person}{Elynn Chen}, {and}
  \bibinfo{person}{Xi Chen}.} \bibinfo{year}{[n.\,d.]}\natexlab{}.
\newblock \showarticletitle{Computation of Optimal MEV in Decentralized
  Exchanges}.
\newblock  (\bibinfo{year}{[n.\,d.]}).
\newblock


\bibitem[Zhang et~al\mbox{.}(2020)]%
        {zhang2020byzantine}
\bibfield{author}{\bibinfo{person}{Yunhao Zhang}, \bibinfo{person}{Srinath
  Setty}, \bibinfo{person}{Qi Chen}, \bibinfo{person}{Lidong Zhou}, {and}
  \bibinfo{person}{Lorenzo Alvisi}.} \bibinfo{year}{2020}\natexlab{}.
\newblock \showarticletitle{Byzantine Ordered Consensus without Byzantine
  Oligarchy}. In \bibinfo{booktitle}{\emph{14th {USENIX} Symposium on Operating
  Systems Design and Implementation ({OSDI} 20)}}. \bibinfo{pages}{633--649}.
\newblock


\end{thebibliography}

\begin{toappendix}
\section{Truncation of Ranked Pairs}
\label{apx:truncation_proof}

Ranked Pairs preserves its output ordering when restricted to considering only an initial segment of its output.
In the statement below, $\mathcal{A}(\cdots)$ denotes Ranked Pairs Voting (Algorithm \ref{alg:rpv}).

\begin{lemma}
\label{lemma:apx_truncation_proof}

Suppose that $(\sigma_1,\ldots,\sigma_n)$ is a set of ordering votes
on a set $V$ of transactions, and let $\sigma=\mathcal{A}(\sigma_1,\ldots,\sigma_n)$.

Consider a set $V^\prime \subset V$ that forms the transactions in an initial segment of $\sigma$,
and for each $i\in [n]$, let $\sigma_i^\prime$ be $\sigma_i$ restricted to $V^\prime$.

Then $\sigma$ extends $\mathcal{A}(\sigma^\prime_1,\ldots,\sigma^\prime_n)$.
\end{lemma}

\begin{proof}

The proof follows by executing $\mathcal{A}(\sigma_1,\ldots,\sigma_n)$ and comparing against an execution
on the restricted input.

Whenever the algorithm includes an edge, inclusion cannot create a cycle in the graph of previously included
edges, so it cannot create a cycle in the graph of previously included edges restricted to $V^\prime$.

Whenever the algorithm rejects an edge $(\tx, \txp)$, it must have included already a directed path from
$\txp$ to $\tx$.  If $\tx$ and $\txp$ both lie within $V^\prime$ (that is, the algorithm run on the restricted input
would iterate over this edge), then the previously included path must be entirely contained within $V^\prime$.  Otherwise,
there would be some $\tx^*$ on this path in $V\setminus V^\prime$, but 
that would imply that $\tx*$ would come before $\tx$ in $\sigma$ and thus that $V^\prime$ is not the set of transactions in
an initial segment of $\sigma$ (as $\tx \in V^\prime$). 
\safeqed
\end{proof}

\section{On the Difficulty of Statically Choosing \texorpdfstring{$\gamma$}{Gamma}}
\label{apx:orderexamples}

One challenge with prior work on this topic is that prior systems require choosing a static
value of $\gamma$ \cite{kelkar2020order,themis,cachin2022quick}.
It is not immediately obvious whether a higher or lower value of $\gamma$ provides a stronger fairness guarantee.
In fact, which value of $\gamma$ gives the stronger guarantee depends on the scenario, as we show by example here.
These examples assume that there are no faulty replicas, and therefore invoke not Definition \ref{defn:gammabatch}
of prior work but the stronger Definition \ref{defn:strongminimal}, which fixes the problems around batch minimality
discussed in \S \ref{sec:minimality}.

For simplicity, in this section, suppose that there are $n=10$ replicas.

\begin{example}[High $\gamma$]
\label{ex:highgamma}
Suppose that $3$ replicas receive $\tx_1 \before \tx_2 \before \tx \before \txp$,
$3$ replicas receive $\tx_2 \before \tx \before \txp \before \tx_1$,
and $4$ replicas receive $\tx \before \txp \before \tx_1 \before \tx_2$.

% 70% 1 < 2
% 70% tx, txp < 1
% 60% 2 < tx,txp 

Then Definition \ref{defn:strongminimal}
only ensures that $\tx \before \txp$ (an ordering relation on which every replica agrees)
if $\gamma \geq 0.6$.

\end{example}

\begin{example}[Low $\gamma$]
\label{ex:lowgamma}
Suppose that $6$ replicas receive $\tx \before \txp$, and $4$ $\txp \before \tx$.
Then Definition \ref{defn:strongminimal}
only ensures that $\tx \before \txp$ in an algorithm's output
if $\gamma\leq 0.6$.
\end{example}

\iffalse
     CC CC 
AAA AAA AA BB
           CC
BBB BBB BB AA

CCC C

Want A < B    8
B < C         6
C < A         6

\fi

\begin{example}[Intermediate $\gamma$]

Suppose that $4$ replicas receive $\tx_1 \before \tx_2 \before \tx_3$,
$4$ replicas receive $\tx_3 \before \tx_1 \before \tx_2$,
and $2$ replicas receive $\tx_2 \before \tx_3 \before \tx_1$.

Then $80\%$ receive $\tx_1\before \tx_2$,
$60\%$ receive $\tx_2 \before \tx_3$,
and $60\%$ receive $\tx_3 \before \tx_1$.
Thus, Definition \ref{defn:strongminimal}
only provides any constraint on the output ordering
(in this case, that the strongest ordering dependency,
between $\tx_1$ and $\tx_2$, is respected)
if $0.6 < \gamma \leq 0.8$.
\end{example}

\section{Ranked Pairs and Interleaved Components}
\label{apx:interleave}

The following example demonstrates the need in Definition \ref{defn:minimalbatch} to allow what would be distinct but incomparable
batches in prior work (specifically Aequitas \cite{kelkar2020order}) to be output in an interleaved ordering, instead of requiring that one
be fully output before the other.
As a corollary, this demonstrates that no process for ordering transactions within Aequitas's batches can always achieve the same ordering
as output in Ranked Pairs.

Consider the following set of ordering votes on $8$ transactions from $16$ replicas.
In this example, we assume no replicas are faulty.

The general pattern is two condorcet cycles of four transactions each, with one set on the extremes and one set in the middle.
Differences from this pattern are bolded for clarify.

\begin{align*}
\tx_1 \before \tx_2  \before \tx_5 \before \tx_6 \before \tx_7 \before \tx_8 \before \tx_3 \before \tx_4 \\
\tx_1 \before \tx_2  \before \tx_6 \before \tx_7 \before \tx_8 \before \tx_5 \before \tx_3 \before \tx_4 \\
\tx_1 \before \tx_2  \before \tx_7 \before \tx_8 \before \tx_5 \before \tx_6 \before \tx_3 \before \tx_4 \\
\tx_1 \before \tx_2  \before \tx_8 \before \tx_5 \before \tx_6 \before \tx_7 \before \tx_3 \before \tx_4 \\
\tx_2 \before \tx_3  \before \tx_5 \before \tx_6 \before \tx_7 \before \mathbf{\tx_4 \before \tx_8} \before \tx_1 \\
\tx_2 \before \tx_3  \before \tx_6 \before \tx_7 \before \tx_8 \before \tx_5 \before \tx_4 \before \tx_1 \\
\tx_2 \before \tx_3  \before \tx_7 \before \tx_8 \before \tx_5 \before \tx_6 \before \tx_4 \before \tx_1 \\
\tx_2 \before \tx_3  \before \tx_8 \before \tx_5 \before \tx_6 \before \tx_7 \before \mathbf{\tx_1 \before \tx_4} \\
\tx_3 \before \tx_4 \before \tx_5 \before \tx_6 \before \tx_7 \before \tx_8 \before \tx_1 \before \tx_2 \\
\tx_3 \before \tx_4 \before \tx_6 \before \tx_7 \before \mathbf{\tx_5 \before \tx_8} \before \tx_1 \before \tx_2 \\
\tx_3 \before \tx_4  \before \tx_7 \before \tx_8 \before \tx_5 \before \tx_6 \before \tx_1 \before \tx_2 \\
\tx_3 \before \tx_4  \before \tx_8 \before \tx_5 \before \tx_6 \before \tx_7 \before \tx_1 \before \tx_2 \\
\tx_4 \before \mathbf{\tx_5 \before \tx_1} \before \tx_6 \before \tx_7 \before \tx_8 \before \tx_2 \before \tx_3 \\
\tx_4 \before \tx_1  \before \tx_6 \before \tx_7 \before \tx_8 \before \tx_5 \before \tx_2 \before \tx_3 \\
\tx_4 \before \tx_1  \before \tx_7 \before \tx_8 \before \tx_5 \before \tx_6 \before \tx_2 \before \tx_3 \\
\tx_4 \before \tx_1  \before \tx_8 \before \tx_5 \before \tx_6 \before \tx_7 \before \tx_2 \before \tx_3
\end{align*}

Sorting edges by weight gives the following:
\begin{enumerate}

	\item[$12/16$] $\tx_1 \before \tx_2$, $\tx_2 \before \tx_3$, $\tx_3 \before \tx_4$, 
					$\tx_5\before \tx_6$, $\tx_6 \before \tx_7$,
					$\tx_7 \before \tx_8$
	\item[$11/16$] $\tx_4 \before \tx_1$, $\tx_8 \before \tx_5$
	\item[$9/16$] $\tx_5 \before \tx_1$, $\tx_4 \before \tx_8$
\end{enumerate}

All others have weight $8/16$ or less.

As such, when Aequitas is run with $\gamma$ between $9/16$ and $11/16$,
would recognize the sets $\lbrace \tx_1, \tx_2, \tx_3, \tx_4\rbrace$ and $\lbrace \tx_5, \tx_6, \tx_7, \tx_8\rbrace$ as incomparable batches,
and would output all of one before beginning to output the other---no matter how it orders transactions within these two batches.

However, Ranked Pairs, on this input, would first select all of the weight $12/16$ edges, reject the weight $11/16$ edges,
and accept the weight $9/16$ edges.
This means that, no matter how the rest of the edges are considered, Ranked Pairs will at least output
$\tx_5 \before \tx_1$, $\tx_1 \before \tx_4$, and $\tx_4 \before \tx_8$.  

This ordering thus cannot be the result of any Aequitas invocation
with this choice of $\gamma$.

\section{Themis and Early Finalization}
\label{apx:early_finalize}

Themis \cite{themis}, by design,
achieves bounded liveness
by acknowledging the possibility that a Condorcet cycle
could be output spread across multiple rounds of its protocol.

For simplicity, assume for the example that the number of faulty replicas $f$ is $0$,
and assume that $\gamma n$ and $\frac{n}{3}$ are integers.  For simplicity of exposition, we assume that $\frac{2}{3} + \frac{1}{n} < \gamma < 1$.
Analogous examples will work for other values of $n$ and $\gamma$.

There will be five transactions in this example.
Ultimately, replicas will receive transactions in the following order.

\begin{align*}
\text{Class 1}~ & ~1 \text{ replica} & \tx_1 \before \tx_2 \before \tx_4 \before \tx_3  \before \tx_5 \\
\text{Class 2}~ & ~(1-\gamma)n \text{ replicas} 				  & \tx_1 \before \tx_2 \before \tx_5 \before \tx_4 \before \tx_3 \\
\text{Class 3}~ & ~n(\gamma - \frac{2}{3}) - 1 \text{ replicas} & \tx_1 \before \tx_2 \before \tx_3 \before \tx_5 \before \tx_4 \\
\text{Class 4}~ & ~\frac{n}{3} \text{ replicas} 				  & \tx_3 \before \tx_1 \before \tx_2 \before \tx_5 \before \tx_4 \\
\text{Class 5}~ & ~\frac{n}{3} \text{ replicas} 				  & \tx_2 \before \tx_3 \before \tx_1 \before \tx_5 \before \tx_4 \\
\end{align*}

Clearly, if Ranked Pairs is run on the full ordering votes, the result would have $\tx_5 \before \tx_4$, as all but $1$ replica
receive those transactions in that order (and that is the first edge considered by Ranked Pairs).

Themis operates in rounds, each of which has three phases.
First, in the ``FairPropose'' phase,
replicas submit ordering votes to a leader, who gathers them and publishes a proposal.
Suppose that, due to the timing of messages over the network,
the leader receives the following set of votes.

\begin{align*}
	\text{Class 1} ~&~ \tx_1 \before \tx_2 \before \tx_4 \before \tx_3 \\
	\text{Class 2} ~&~ \tx_1 \before \tx_2 \before \tx_5 \before \tx_4 \before \tx_3\\
	\text{Class 3} ~&~ \tx_1 \before \tx_2 \before \tx_3\\
	\text{Class 4} ~&~ \tx_3 \before \tx_1 \before \tx_2\\
	\text{Class 5} ~&~ \tx_2 \before \tx_3 \before \tx_1
\end{align*}

Themis then considers transactions by the number of replicas which include the transactions in their votes.
Those in $n$ votes are ``solid'', those with fewer than $n$ but at least $\gamma n + 1$ are ``shaded,''
and the rest are ``blank.''  Clearly transactions $\tx_1,\tx_2,$ and $\tx_3$ are solid,
$\tx_4$ is shaded, and $\tx_5$ is blank.

Themis then builds a dependency graph between the solid and shaded vertices.  
A dependency is added between two transactions $\tx$ and $\txp$ if at least $(1-\gamma)n+1$ replicas
vote $\tx\before \txp$ and fewer replicas vote $\txp \before \tx$.

As such, the dependencies that get added here 
are $\tx_1 \before \tx_2$, $\tx_2 \before \tx_3$, $\tx_3 \before \tx_1$,
$\tx_1 \before \tx_4$, $\tx_2 \before \tx_4$ and $\tx_4 \before \tx_3$.

The next step is to compute the strongly connected components of this graph and topologically sort it.
In this case, the graph has one component.
This component is the last in the ordering that contains a solid vertex,
so, per Themis's rules, this graph is the output of FairPropose.

Next, replicas send ``updates'' to a ``FairUpdate'' algorithm, which orders shaded transactions in the output of FairPropose relative
to the other solid and shaded transactions in that output.
This is where applying Ranked Pairs (or any ranking algorithm) to each output batch of Themis individually breaks down.
The response of honest replicas at this point (namely, those in classes 3, 4 and 5) is to vote that they received $\tx_4$ after
each of $\tx_1,\tx_2$ and $\tx_3$.  In fact, no edges get added to the dependency graph at this stage,
because the graph is already complete.

Finally, in ``FairFinalize,'' Themis computes some ordering on transactions $\tx_1, \tx_2,\tx_3,$ and $\tx_4$, and includes them in the output.

Observe that $\tx_5$ is forced to come after $\tx_4$.

\section{Algorithm \ref{alg:stream_rpv} Is Not Asymptotically Live}
\label{apx:non_live}

Consider the following setting.  There are two replicas (neither of which is faulty),
and a countably infinite set of transactions $\lbrace \tx_1,\tx_2,\tx_3,\ldots\rbrace$.

\begin{example}~\\
\label{ex:badorder}
Replica $1$ receives transactions in the order \\
$\tx_2 \before \tx_1 \before \tx_4 \before \tx_3 \before \tx_6 \before \tx_5 \before \ldots \before \tx_{2*i} \before \tx_{2*i - 1} \before \ldots$, for all $i \geq 1$.\\
Replica $2$ receives transactions in the order \\
$\tx_1 \before \tx_3 \before \tx_2 \before \tx_5 \before \tx_4 \before \tx_7 \before \tx_6 \ldots \before \tx_{2*i+1} \before \tx_{2*i} \ldots$,
for all $i\geq 1$.

This means that any edge $(\tx_i, \tx_j)$ has weight $\frac{1}{2}$ if $i=j+1$ or $j=i+1$,
and otherwise $1$ (resp. $0$) if $i<j$ (resp. $i>j$).

Additionally, suppose that of the edges of weight $1/2$ of the form $e_i=(\tx_{i+1}, \tx_{i})$, the deterministic tiebreaking ordering
sorts $(\tx_{i+1}, \tx_{i})$ just before $(\tx_{i}, \tx_{i+1})$ and $e_i$ before $e_j$ for $i>j$.
\end{example}

\begin{theorem}
Algorithm \ref{alg:stream_rpv} never outputs any transaction on any finite subset of the input in Example \ref{ex:badorder}.
\end{theorem}

\begin{proof}
On any finite truncation of this input,
suppose that $N$ is the highest index such that $\tx_N$ and all $\tx_i$ for $i\leq N$ appear in the votes of both replicas.

Then $\tx_{N+1}$ must be present in the output of one of the replicas, but not both.
This means that the other replica's vote (that does not contain $\tx_{N+1})$ may contain $\tx_{N+2}$, 
but the transaction after this in the replica's true vote is $\tx_{N+1}$, so the other replica's vote must stop here.
As such,
the streamed ordering graph used in Algorithm \ref{alg:stream_rpv} will contain all transactions up to $\tx_N$,
possibly $\tx_{N+2}$, and $\hat{v}$.  There will be one edge from $\hat{v}$ to $\tx_N$, and if it is present, an edge from $\hat{v}$ to $\tx_{N+2}$.

Algorithm \ref{alg:stream_rpv} will accept all of the edges of 
weight $1$ and then proceed to the edges of weight $\frac{1}{2}$.  The first to be considered
will be $e_{N}$, which will be marked \textit{indeterminate}
because of the path $(\tx_{N-1}, \hat{v}),(\hat{v}, \tx_{N})$
(and so will the reversal of $e_N$ also be marked \textit{indeterminate}).
Continuing through the sequence of edges,
If $e_k$ is marked \textit{indeterminate},
then so too will $e_{k-1}$, because of the path
$(\tx_{k-2}, \tx_{k}), e_k$.
Note that this is the only path that Algorithm \ref{alg:stream_rpv} considers when visiting $e_{k-1}$,
as the only other transactions in $U_{\tx_{k-1}, \tx_{k-2}}$ are $\tx_{k}$ and $\tx_{k-3}$.

Thus, Algorithm \ref{alg:stream_rpv} never outputs any transactions on a finite subset of the input in Example \ref{ex:badorder}.
\end{proof}

We remark that this construction requires that the ordering between edges of weight $\frac{1}{2}$ is not isomorphic to a subset of $\omega$.

\section{Oracle Implementation}
\label{sec:oracle_impl}

Algorithm \ref{alg:stream_rpv_2} relies on an oracle to a past invocation of the algorithm.
Implementing the oracle naively would require maintaining a map from the set of all edges of the past invocation
to their status, the cost of which would grow logarithmically.
However, the only edges for which a path search must be performed are those $(\tx, \txp)$ for which neither $\tx$ nor $\txp$
were included in the prior output (if either $\tx$ or $\txp$ were in the prior output,
then monotonicity immediately implies whether the edge is included or excluded).
As such, the map in the oracle could be pruned of edges connected to transactions 
in the prior output. Although the number of such edges is not bounded, it is not guaranteed to grow over time without bound.

\section{Algorithms}
\label{apx:algorithms}

We include for completeness Algorithm \ref{alg:rpv}, which reproduces the (non-streaming)
 Ranked Pairs algorithm of \citet{tideman1987independence}.

\begin{algorithm}
\caption{Ranked Pairs Voting}
\label{alg:rpv}
\KwIn{An ordering vote from each replica $\lbrace \sigma_i \rbrace_{i\in [n]}$}% on transactions $(\tx_1,...,\tx_k)$}
$G=(V,E,w) \gets G(\sigma_1,...,\sigma_n)$ \aaaielide{\tcc*[r]{Compute ordering graph}}{\\}
$H \gets (V,\emptyset)$ \aaaielide{\tcc*[r]{Initialize empty graph on same vertex set}}{\\}
Sort $E$ by edge weight\\
\ForEach{$e\in E$}{
	Add $e$ to $H$ if it would not create a directed cycle in $H$
}
\KwOut{The topological sort of $H$}
\aaaielide{\tcc*[r]{The loop considered every edge, so the topological sort of $H$ is unique}}{}
\end{algorithm}

Algorithm \ref{alg:stream_rpv} forms the first step to adapting Ranked Pairs to the streaming setting.
However, its reliance on a deterministic tiebreaking rule means that it cannot guarantee asymptotic liveness.

\begin{algorithm}
\caption{Streaming Ranked Pairs, Preliminary Version}
\label{alg:stream_rpv}
\KwIn{An ordering vote from each replica $\lbrace \sigma_i \rbrace_{i\in [n]}$}
\aaaielide{
$\hat{G}=(V,E,w) \gets \hat{G}(\sigma_1,...,\sigma_n)$ \tcc*[r]{Compute streamed ordering graph}
$H \gets (V,\emptyset)$ \tcc*[r]{Initialize empty graph on same vertex set}}
{
  $\hat{G}=(V,E,w) \gets \hat{G}(\sigma_1,...,\sigma_n)$, $H\gets (V, \emptyset)$\\
}
\ForEach{edge $(\tx, \hat{v})$ and $(\hat{v}, \tx)$ with $w(e)=1$}{
  Add $e$ to $H$, and mark it as \textit{indeterminate}
}
Sort $E$ by edge weight, with a fixed tiebreaking rule.\\
\ForEach{$e=(\tx_i, \tx_j)\in E$}{
  $U_{\tx_i, \tx_j} \gets (V\setminus (R_{\tx_i}\cup P_{\tx_j}))\cup \lbrace \hat{v}\rbrace$\\
  \aaaielide{\tcc*[r]{Recall that $\hat{v}$ is the ``future'' vertex of the streamed ordering graph}}{}
  If there is a directed path in $H$ restricted to $U_{\tx_i,\tx_j}$
    from $\tx_j$ to $\tx_i$
    where every edge is \textit{determinate}, then do not include the edge in $H$.\\
  If there is no directed path in $H$ restricted to $U_{\tx_i,\tx_j}$
    from $\tx_j$ to $\tx_i$
    where every edge is either \textit{determinate} or \textit{indeterminate},
    add $(\tx_i, \tx_j)$ to $H$ and mark it as \textit{determinate}.\\
  Otherwise, add $(\tx_i, \tx_j)$ to $H$ and mark it as \textit{indeterminate}.
}
\KwOut{The topological sort of $H$ (with sorting comparisons restricted to \textit{determinate} edges), 
up to but not including the first transaction $\tx$ bordering an \textit{indeterminate}
edge}
\end{algorithm}

\section{Protocol Instantiation}
\label{sec:instantiation}
The algorithms thus far have assumed that every replica will eventually vote on every transaction.
Any faulty replica, however, could exclude a transaction from its output, leaving the transaction
(and the protocol) permanently stuck.

The main idea to resolve this difficulty
is that if a replica refuses to vote on a transaction, then that replica must be faulty.
If a replica is faulty, then it suffices for the rest of the replicas to make up a vote on its behalf.

We give here an example method for instantiating our sequencing protocol on top of existing protocol components,
although this is far from the only method.

As a base layer, we require some consensus protocol, such as \cite{yin:hotstuff,lokhava:stellar,castro1999practical,miller2016honey}, 
that proceeds in \textit{rounds}.
At any time, every replica can submit a continuation of its previous ordering; that is, an ordered list of transactions
that extends its ordering vote.  The consensus protocol then comes to consensus on a set of received ordering votes.
The Streaming Ranked Vote algorithm is run (by every replica)
on the set of received votes after every new set of votes is finalized by the consensus protocol.
Assuming correctness of the consensus protocol, every replica will run the algorithm on the same input,
and thereby produce the same output ordering.

We assume that every replica has a registered public key. 
Every vote continuation must be cryptographically signed by the associated secret key.
Furthermore, every continuation must be accompanied by the hash of the previously submitted continuation,
so that if one continuation vote from a replica is lost (perhaps by malicious activity of consensus participants)
the replica's ordering cannot be manipulated.

We make a network synchrony assumption, akin to Assumption \ref{ass:sync}.
If a transaction appears first in some replica's vote in round $T$ of the consensus protocol,
then we assume that there is some known bound $k$ such that every other honest replica will have time to 
receive and submit a vote on the transaction before the end of round $T+k$.
As such, after every round $r$, the protocol computes the set of all transactions that have not received
votes from every replica but were voted on by some replica at or before round $r-k$. For each of these transactions $\tx$,
the protocol appends a vote for $\tx$ at the end of a replica's current vote sequence if $\tx$ is not already present.
When appending multiple transactions to a replica's vote, transactions can be sorted deterministically (i.e. by hash).

This protocol guarantees that every transaction will receive a vote from all replicas within at most $k$ rounds
after it is first submitted to the protocol.  
In this setting, $k$ must be configured in advance by the protocol.  The parameter must be set so that every honest replica is able
 to submit a vote on a transaction within $k$ rounds.  Naturally, the choice of $k$ depends on
 network conditions
the censorship-resistance (or ``chain-quality'' \cite{garay2015bitcoin}) of the consensus protocol.

As an example alternative, one could also use the protocol in \citet{kelkar2020order}, which is based on atomic broadcast and a ``Set Byzantine Agreement'' primitive.
This protocol waits until all honest replicas are able to add a transaction $\tx$
 to their ordering votes,
and then invokes a protocol to agree on the set of replicas $U_{\tx}$ which have the transaction in their ordering votes.
This set may not be the set of all replicas, but includes all honest replicas.  To use our ordering protocol on this consensus protocol,
it suffices to make up a vote for the (presumed faulty) replicas not in $U_{\tx}\cap U_{\txp}$.  
Replicas in $U_{\tx} \setminus U_{\txp}$ 
can be presumed to vote $\tx \before \txp$ and vice versa, and replicas not in $U_{\tx}\cup U_{\txp}$ can be presumed to vote (arbitrarily) by, e.g., comparing deterministic hashes of $\tx$ and $\txp$.

\end{toappendix}

\end{document}